\documentclass[11pt]{article}

\usepackage{array,multirow}
\usepackage[letterpaper,hmargin=1in,vmargin=1.25in]{geometry}
\usepackage{amsmath}
\usepackage{array,fullpage,multirow}
\usepackage{amssymb,amsthm,sectsty,url,nicefrac}
\usepackage[usenames,dvipsnames]{xcolor}
\usepackage[colorlinks=true,linkcolor=NavyBlue,citecolor=teal]{hyperref}
\usepackage{cleveref}
\usepackage{aliascnt}
\usepackage{braket}
\usepackage{boxedminipage}
\usepackage[boxed]{algorithm}
\usepackage{algpseudocode}
\usepackage{tikz}
\usetikzlibrary{arrows,positioning} 

\theoremstyle{nicetheorem}

\newtheorem{theorem}{Theorem}[section]
\newtheorem*{theorem*}{Theorem}
\newtheorem*{lemma*}{Lemma}

\newaliascnt{lemma}{theorem}
\newtheorem{lemma}[lemma]{Lemma}
\aliascntresetthe{lemma}
\crefname{lemma}{Lemma}{Lemmas}

\newaliascnt{problem}{theorem}

\aliascntresetthe{problem}
\crefname{problem}{Problem}{Problems}

\newaliascnt{claim}{theorem}
\newtheorem{claim}[claim]{Claim}
\aliascntresetthe{claim}
\crefname{claim}{Claim}{Claims}

\newaliascnt{corollary}{theorem}
\newtheorem{corollary}[corollary]{Corollary}
\aliascntresetthe{corollary}
\crefname{corollary}{Corollary}{Corollaries}

\newaliascnt{construction}{theorem}

\aliascntresetthe{construction}
\crefname{construction}{Construction}{Constructions}

\newaliascnt{fact}{theorem}

\aliascntresetthe{fact}
\crefname{fact}{Fact}{Facts}

\newaliascnt{proposition}{theorem}

\aliascntresetthe{proposition}
\crefname{proposition}{Proposition}{Propositions}

\newaliascnt{conjecture}{theorem}

\aliascntresetthe{conjecture}
\crefname{conjecture}{Conjecture}{Conjectures}

\newaliascnt{definition}{theorem}
\newtheorem{definition}[definition]{Definition}
\aliascntresetthe{definition}
\crefname{definition}{Definition}{Definitions}

\newaliascnt{remark}{theorem}

\aliascntresetthe{remark}
\crefname{remark}{Remark}{Remarks}

\newaliascnt{observation}{theorem}

\aliascntresetthe{observation}
\crefname{observation}{Observation}{Observations}

\crefname{algorithm}{Algorithm}{Algorithms}

\newcommand{\abs}[1]{\left|#1\right|}

\newcommand\E{\mathop{\mathbb E}}
\newcommand{\U}{\mathbf U}
\newcommand{\N}{\mathbb N}
\newcommand{\R}{\mathbb R}

\newcommand{\B}{\{ 0,1 \}}

\newcommand{\BPP}{{\mathsf {BPP}}}
\newcommand{\NP}{{\mathsf {NP}}}
\newcommand{\NC}{{\mathsf {NC}}}
\newcommand{\AM}{{\mathsf {AM}}}
\newcommand{\coAM}{{\textsf{co}\AM}}
\newcommand{\MA}{{\mathsf {MA}}}
\newcommand{\mNP}{\mathsf{m}{\NP}}

\newcommand{\coNP}{{\textsf{co}\NP}}

\newcommand{\ioBPP}{{\textsf{io-}\BPP}}
\newcommand{\PP}{{\mathsf {P}}}

\newcommand{\eps}{\varepsilon}

\newcommand{\CNF}{3\mathsf{CNF}}

\newcommand{\io}{{i\mathcal{O}}}

\newcommand{\negl}{{\mathsf{neg}}}
\newcommand{\secret}{{S}}
\newcommand{\RECON}{{\mathsf{RECON}}}
\newcommand{\SETUP}{{\mathsf{SETUP}}}
\newcommand{\Sc}{{\mathcal{S}}}
\newcommand{\Com}{{\mathsf{Com}}}
\newcommand{\com}{{\mathsf{c}}}

\newcommand{\shares}{{\Pi}}
\newcommand{\share}{\shares}

\newcommand{\parties}{{\mathcal P}}
\newcommand{\party}{{\mathsf p}}
\newcommand{\ver}{{V}}
\newcommand{\sop}{{X}} %subset of parties
\newcommand{\Samp}{{\mathsf{Samp}}}
\newcommand{\resd}{{\mathsf{resD}}}
\newcommand{\resm}{{\mathsf{bias}}}
\newcommand{\bias}{{\mathsf{bias}}}
\newcommand{\Mest}{\mathsf{B}}
\newcommand{\Dver}{\mathsf{D}_{\mathsf{ver}}}

\newcommand{\myand}{{\;\wedge\;}}

\newcommand{\ct}{\mathsf{ct}}
\newcommand{\Encrypt}{\mathsf{Encrypt}}
\newcommand{\Decrypt}{\mathsf{Decrypt}}
\newcommand{\poly}{\mathsf{poly}}

\newcommand{\mono}{\mathsf{mono}}
\newcommand{\crs}{\mathsf{crs}}
\def\({\left(}
\def\){\right)}

\makeatletter \renewenvironment{proof}[1][\proofname]{%
  \par\pushQED{\qed}\normalfont%
  \topsep6\p@\@plus6\p@\relax
  \trivlist\item[\hskip\labelsep\bfseries#1\@addpunct{.}]%
  \ignorespaces }{%
  \popQED\endtrivlist\@endpefalse } \makeatother

\makeatletter
\newtheorem*{rep@theorem}{\rep@title}
\newcommand{\newreptheorem}[2]{%
  \newenvironment{rep#1}[1]{%
    \def\rep@title{#2 \ref{##1}}%
    \begin{rep@theorem}}%
    {\end{rep@theorem}}}
\makeatother

\newreptheorem{theorem}{Theorem}

\makeatletter
\newtheorem*{rep@lemma}{\rep@title}
\newcommand{\newreplemma}[2]{%
  \newenvironment{rep#1}[1]{%
    \def\rep@title{#2 \ref{##1}}%
    \begin{rep@lemma}}%
    {\end{rep@lemma}}}
\makeatother

\newreplemma{lemma}{Lemma}

\newcommand{\aka} {also known as\ }
\newcommand{\resp}{resp.,\ }
\newcommand{\ie}  {i.e.,\ }
\newcommand{\eg}  {e.g.,\ }
\newcommand{\etal}{{et~al.\ }}
\newcommand{\etalcite}[1]{{et~al.~\cite{#1}}}
\newcommand{\uniran}{{\stackrel{\mathsf{R}}{\leftarrow}}}

\newcommand{\BAD}{\mathsf{BAD}}

\newcommand{\Prp}[1]{\Pr\left[#1\right]}

\makeindex

\title{Secret-Sharing for $\NP$}

\author{Ilan Komargodski\footnote{Weizmann Institute of Science. Email: {\tt
      \{ilan.komargodski,moni.naor,eylon.yogev\}@weizmann.ac.il}. Research
    supported in part by a grant from the Israel Science Foundation, the I-CORE
    Program of the Planning and Budgeting Committee, BSF and the Israeli
    Ministry of Science and Technology. Moni Naor is the incumbent of the Judith
    Kleeman Professorial Chair.}%\\
  \and Moni Naor\footnotemark[1]%\\
  \and
  Eylon Yogev\footnotemark[1]%\\
}

\date{}
\begin{document}

\maketitle

\begin{abstract}
  A computational secret-sharing scheme is a method that enables a dealer, that
  has a secret, to distribute this secret among a set of parties such that a
  ``qualified'' subset of parties can efficiently reconstruct the secret while
  any ``unqualified'' subset of parties cannot efficiently learn anything
  about the secret. The collection of ``qualified'' subsets is defined by a
  monotone Boolean function.

  It has been a major open problem to understand which (monotone) functions can
  be realized by a computational secret-sharing scheme. Yao suggested a method
  for secret-sharing for any function that has a polynomial-size monotone
  circuit (a class which is strictly smaller than the class of monotone
  functions in $\PP$). Around 1990
  Rudich raised the possibility of obtaining secret-sharing for all monotone
  functions in $\NP$: In order to reconstruct the secret a set of parties must
  be ``qualified'' and provide a witness attesting to this fact.

  Recently, Garg~\etal (STOC 2013) put forward the concept of
  \textsf{witness encryption}, where the goal is to encrypt a message relative
  to a statement $x\in L$ for a language $L\in\NP$ such that anyone holding a
  witness to the statement can decrypt the message, however if $x\notin L$,
  then it is computationally hard to decrypt. Garg~\etal showed how to construct
  several cryptographic primitives from witness encryption and gave a candidate
  construction.

  One can show that computational secret-sharing implies witness encryption for
  the same language. Our main result is the converse: we give a construction of
  a computational secret-sharing scheme for \emph{any} monotone function in
  $\NP$ assuming witness encryption for $\NP$ and one-way functions. As a
  consequence we get a completeness theorem for secret-sharing: computational
  secret-sharing scheme for any \emph{single} monotone $\NP$-complete function
  implies a computational secret-sharing scheme for \emph{every} monotone
  function in $\NP$.
\end{abstract}

\newpage
\tableofcontents
\newpage

\section{Introduction}

A \textsf{secret-sharing scheme} is a method that enables a dealer, that has a
secret piece of information, to distribute this secret among $n$ parties such
that a ``qualified'' subset of parties has enough information to reconstruct the
secret while any ``unqualified'' subset of parties learns nothing about the
secret. A monotone collection of ``qualified'' subsets (\ie subsets of parties
that can reconstruct the secret) is known as an \textsf{access structure}, and
is usually identified with its characteristic monotone function.\footnote{It is
  most sensible to consider only \emph{monotone} sets of ``qualified'' subsets
  of parties. A set $M$ of subsets is called monotone if $A\in M$ and
  $A\subseteq A'$, then $A'\in M$. It is hard to imagine a meaningful method for
  sharing a secret to a set of ``qualified'' subsets that does not satisfy this
  property.}  Besides being interesting in their own right, secret-sharing
schemes are an important building block in many cryptographic protocols,
especially those involving some notion of ``qualified'' sets (\eg multi-party
computation, threshold cryptography and Byzantine agreement).  For more
information we refer to the extensive survey of Beimel on secret-sharing schemes
and their applications \cite{Beimel11}.

A significant goal in constructing secret-sharing schemes is to \emph{minimize}
the amount of information distributed to the parties.  We say that a
secret-sharing scheme is \emph{efficient} if the size of all shares is
polynomial in the number of parties and the size of the secret.

Secret-sharing schemes were introduced in the late 1970s by Blakley
\cite{Blakley79} and Shamir \cite{Shamir79} for the {\em threshold access
  structure}, \ie where the subsets that can reconstruct the secret are all the
sets whose cardinality is at least a certain threshold. Their constructions were
fairly efficient both in the size of the shares and in the computation required
for sharing and reconstruction. Ito, Saito and Nishizeki \cite{ISN93} considered
general access structures and showed that every monotone access structure has a
(possibly \emph{inefficient}) secret-sharing scheme that realizes it. In their
scheme the size of the shares is proportional to the DNF (resp.\ CNF) formula
size of the corresponding function. Benaloh and Leichter \cite{Leichter88}
proved that if an access structure can be described by a polynomial-size
monotone \emph{formula}, then it has an efficient secret-sharing scheme. The
most general class for which secret-sharing is known was suggested by Karchmer
and Wigderson~\cite{KarchmerW93} who showed that if the access structure can be
described by a polynomial-size monotone {\em span program} (for instance,
undirected connectivity in a graph), then it has an efficient secret-sharing
scheme. Beimel and Ishai~\cite{BeimelI05} proposed a secret-sharing scheme for
an access structure which is conjectured to lie outside $\NC$. On the other
hand, there are no known lower bounds that show that there exists an access
structure that requires only inefficient secret-sharing
schemes.\footnote{Moreover, there are not even non-constructive lower bounds for
  secret-sharing schemes. The usual counting arguments (\eg arguments that show
  that most functions require large circuits) do not work here since one needs
  to enumerate over the sharing and reconstruction algorithms whose complexity
  may be larger than the share size.}

\paragraph{Computational Secret-Sharing.}
In the secret-sharing schemes considered above the security is guaranteed
information theoretically, that is, even if the parties are computationally
unbounded. These secret-sharing schemes are known as \textsf{perfect
  secret-sharing schemes}. A natural variant, known as \textsf{computational
  secret-sharing schemes}, is to allow only computationally limited dealers and
parties, \ie they are probabilistic algorithms that run in polynomial-time.
More precisely, a computational secret-sharing scheme is a secret-sharing scheme
in which there exists an \emph{efficient} dealer that generates the shares such
that a ``qualified'' subset of parties can \emph{efficiently} reconstruct the
secret, however, an ``unqualified'' subset that pulls its shares together but
has only limited (\ie polynomial) computational power and attempts to
reconstruct the secret should fail (with high
probability). Krawczyk~\cite{Krawczyk93} presented a computational
secret-sharing scheme for threshold access structures that is more efficient (in
terms of the size of the shares) than the perfect secret-sharing schemes given
by Blakley and Shamir~\cite{Blakley79,Shamir79}. In an unpublished work
(mentioned in \cite{Beimel11}, see also Vinod~\etalcite{VinodNSRK03}), Yao
showed an efficient computational secret-sharing scheme for access structures
whose characteristic function can be computed by a polynomial-size monotone
\emph{circuit} (as opposed to the \emph{perfect} secret-sharing of Benaloch and
Leichter~\cite{Leichter88} for polynomial-size monotone {\em formulas}). Yao's
construction assumes the existence of pseudorandom generators, which can be
constructed from any one-way function \cite{HastadILL99}. There are access
structures which are known to have an efficient \emph{computational}
secret-sharing schemes but are not known to have efficient \emph{perfect}
secret-sharing schemes, \eg directed connectivity.\footnote{In the access
  structure for directed connectivity, the parties correspond to edge slots in
  the complete \emph{directed} graph and the ``qualified'' subsets are those
  edges that connect two distinguished nodes $s$ and $t$.} % It is known that
% directed connectivity has a polynomial-size monotone circuit of depth $O(\log
% n)$; this circuit has unbounded fan-in (that is, Yao's scheme gives an
% efficient computational secret-sharing scheme). This implies a monotone
% formula of size $n^{O(\log n)}$ (that is, Benaloh and Leichter's scheme
% \cite{Leichter88} gives an \emph{inefficient} perferct secret-sharing
% scheme).
Yao's scheme does not include all monotone access structures with an efficient
algorithm to determine eligibility. One notable example where no efficient
secret-sharing is known is matching in a graph.\footnote{\label{fn:matching}In
  the access structure for matching the parties correspond to edge slots in the
  complete graph and the ``qualified'' subsets are those edges that
  \emph{contain} a perfect matching. Even though matching is in $\PP$, it is
  known that there is no monotone circuit that computes it \cite{Razborov85}.}
Thus, a major open problem is to answer the following question:

\begin{center}
\emph{Which access structures have
efficient computational secret-sharing schemes, and what cryptographic
assumptions are required for that?}
\end{center}

\paragraph{Secret-Sharing for $\NP$.}
Around 1990 Steven Rudich raised the possibility of obtaining secret-sharing
schemes for an even more general class of access structures than $\PP$: monotone
functions in $\NP$, \aka $\mNP$.\footnote{Rudich raised it in private
  communication with the second author around 1990 and was not written to the
  best of our knowledge; some of Rudich's results can be found
  in Beimel's survey~\cite{Beimel11} and in Naor's
  presentation~\cite{NaorSlides}.} An access structure that is defined by a
function in $\mNP$ is called an $\mNP$ access structure. Intuitively, a
secret-sharing scheme for an $\mNP$ access structure is defined (in the natural
way) as following: for the ``qualified'' subsets there is a witness attesting to
this fact and \emph{given} the witness it should be possible to reconstruct the
secret. On the other hand, for the ``unqualified'' subsets there is no witness,
and so it should not be possible to reconstruct the secret.  For example,
consider the Hamiltonian access structure. In this access structure the parties
correspond to edges of the complete undirected graph, and a set of parties
$\sop$ is said to be ``qualified'' if and only if the corresponding set of edges
contains a Hamiltonian cycle and the set of parties knows a witness attesting to
this fact.

Rudich observed that if $\NP\neq \coNP$, then there is no \emph{perfect}
secret-sharing scheme for the Hamiltonian access structure in which the sharing
of the secret can be done efficiently (\ie in
polynomial-time).\footnote{Moreover, it is possible to show that if
  $\NP\not\subseteq \coAM$, then there is no \emph{statistical} secret-sharing
  scheme for the Hamiltonian access structure in which the sharing of the secret
  can be done efficiently \cite{NaorSlides}.} This (conditional) impossibility
result motivates looking for \emph{computational} secret-sharing schemes for the
Hamiltonian access structure and other $\mNP$ access structures. Furthermore,
Rudich showed that the construction of a computational secret-sharing schemes
for the Hamiltonian access structure gives rise to a protocol for oblivious
transfer. More precisely, Rudich showed that if one-way functions exist and
there is a \emph{computational} secret-sharing scheme for the Hamiltonian access
structure (\ie with efficient sharing and reconstruction), then efficient
protocols for oblivious transfer exist.\footnote{The resulting reduction is
  \emph{non}-black-box. Also, note that the results of Rudich apply for any
  other monotone $\NP$-complete problem as well.} In particular, constructing a
computational secret-sharing scheme for the Hamiltonian access structure
assuming one-way functions will resolve a major open problem in cryptography and
prove that Minicrypt$=$Cryptomania, to use Impagliazzo's terminology
\cite{Impagliazzo95}.

In the decades since Rudich raised the possibility of access structures beyond
$\PP$ not much has happened. This changed with the work on \textsf{witness
  encryption} by Garg~\etalcite{GargGSW13}, where the goal is to encrypt a
message relative to a statement $x\in L$ for a language $L\in\NP$ such that:
Anyone holding a witness to the statement can decrypt the message, however, if
$x\notin L$, then it is computationally hard to decrypt. Garg~\etal showed how
to construct several cryptographic primitives from witness encryption and gave
a candidate construction.

A by-product of the proposed construction of Garg~\etal was a construction of a
computational secret-sharing scheme for a \emph{specific} monotone
$\NP$-complete language. However, understanding whether one can use a
secret-sharing scheme for any single (monotone) $\NP$-complete language in order
to achieve secret-sharing schemes for any language in $\mNP$ was an open
problem. One of our main results is a positive answer to this question. Details follow.

%It is
%relatively simple to show that secret-sharing for an $\NP$ language $L$ implies
%witness encryption for the statement $x \in L$ (see
%Garg~et~al.).  

\paragraph{Our Results.}
In this paper, we construct a secret-sharing scheme for \emph{every} $\mNP$
access structure assuming witness encryption for $\NP$ and one-way functions. In
addition, we give two variants of a formal definition for secret-sharing for
$\mNP$ access structures (indistinguishability and semantic security) and prove
their equivalence.

\begin{theorem}\label{thm:main}
  Assuming witness encryption for $\NP$ and one-way functions, there is an
  efficient computational secret-sharing scheme for every $\mNP$ access
  structure.
\end{theorem}

We remark that if we relax the requirement of computational secret-sharing such
that a ``qualified'' subset of parties can reconstruct the secret with very high
probability (say, negligibly close to 1), then our scheme from
\Cref{thm:main} actually gives a secret-sharing scheme for every monotone
functions in $\MA$.

As a corollary, using the fact that a secret-sharing scheme for a language
implies witness encryption for that language and using the completeness of
witness encryption,\footnote{\label{fn:comp_we}Using standard Karp/Levin
  reductions between $\NP$-complete languages, one can transform a witness
  encryption scheme for a single $\NP$-complete language to a witness encryption
  scheme for any other language in $\NP$.} we obtain a completeness theorem for
secret-sharing.
\begin{corollary}[Completeness of Secret-Sharing]\label{cor:main}
  Let $L$ be a monotone language that is $\NP$-complete (under Karp/Levin
  reductions) and assume that one-way functions exist. If there exists a
  computational secret-sharing scheme for the access structure defined by $L$,
  then there are computational secret-sharing schemes for \emph{every} $\mNP$
  access structure.
\end{corollary}

\subsection{On Witness Encryption and Its Relation to Obfuscation}
Witness encryption was introduced by Garg~\etalcite{GargGSW13}. They gave a
formal definition and showed how witness encryption can be combined with other
cryptographic primitives to construct public-key encryption (with efficient key
generation), identity-based encryption and attribute-based encryption. Lastly,
Garg~\etal presented a candidate construction of a witness encryption scheme
which they assumed to be secure. In a more recent work, a new construction of a
witness encryption scheme was proposed by Gentry, Lewko and Waters
\cite{GentryLW14}.

Shortly after the paper of Garg~\etalcite{GargGSW13} a candidate construction of
\textsf{indistinguishability obfuscation} was proposed by
Garg~\etalcite{GargGH0SW13}. An indistinguishability obfuscator is an algorithm
that guarantees that if two circuits compute the same function, then their
obfuscations are computationally indistinguishable. The notion of
indistinguishability obfuscation was originally proposed in the seminal work of
Barak~\etalcite{BarakGIRSVY01,BarakGIRSVY12}.

Recently, there have been two significant developments regarding
indistinguishability obfuscation: first, candidate constructions for obfuscators
for all polynomial-time programs were proposed
\cite{GargGH0SW13,BrakerskiR14,BarakGKPS14,PassST14,GentryLSW14} and second,
intriguing applications of indistinguishability obfuscation when combined with
other cryptographic primitives\footnote{See \cite{KomargodskiMNPRY14} for a
  thorough discussion of the need in additional hardness assumptions on top of
  $\io$.} have been demonstrated (see, \eg
\cite{GargGH0SW13,SahaiW14,BonehZ14}).

As shown by Garg~\etalcite{GargGH0SW13}, indistinguishability obfuscation
implies witness encryption for all $\NP$, which, as we show in
\Cref{thm:main}, implies secret-sharing for all $\mNP$. In fact, using the
completeness of witness encryption (see Footnote~\ref{fn:comp_we}), even an
indistinguishability obfuscator for $\CNF$ formulas (for which there is a
simple candidate construction \cite{BrakerskiR14a}) implies
witness encryption for all $\NP$.  Understanding whether witness encryption is
strictly weaker than indistinguishability obfuscation is an important open
problem.

A summary of the known relations between the above mentioned objects can be
found in \Cref{fig:zoo}.

\subsection{Other Related Work}
A different model of secret-sharing for $\mNP$ access structures was suggested
by Vinod~\etalcite{VinodNSRK03}. Specifically, they relaxed the requirements
of secret-sharing by introducing a semi-trusted third party $T$ who is allowed
to interact with the dealer and the parties. They require that $T$ does not
learn anything about the secret and the participating parties. In this model,
they constructed an efficient secret-sharing scheme for any $\mNP$ access
structures (that is also efficient in terms of the round complexity of the
parties with $T$) assuming the existence of efficient oblivious transfer
protocols.

\subsection{Main Idea}
Let $\Com$ be a perfectly-binding commitment scheme. Let $M \in \mNP$ be an
access structure on $n$ parties $\parties=\{\party_1,\dots,\party_n\}$. Define
$M'$ to be the $\NP$ language that consists of sets of $n$ strings
$\com_1,\dots,\com_n$ as follows. $M'(\com_1,\dots,\com_n) = 1$ if and only if
there exist $r_1,\dots,r_n$ such that $M(x) = 1$, where $x=x_1\dots x_n$ is such
that
\begin{align*}
  \forall i\in[n]:\; \quad x_i =
  \begin{cases}
    1 & \text{if } r_i \neq \bot \text{ and }\Com(i, r_i) =
    \com_i,\\
    0 & \text{otherwise.}
  \end{cases}
\end{align*}

For the language $M'$ denote by $(\Encrypt_{M'}, \Decrypt_{M'})$ the witness
encryption scheme for $M'$. A secret-sharing scheme for the access structure 
$M$ consists of a setup phase in which the dealer distributes secret shares to 
the parties. First, the dealer samples uniformly at random $n$ openings 
$r_1,\dots,r_n$. Then, the dealer computes a witness encryption $\ct$ of the 
message $\secret$ with respect to the instance $\(\com_1= 
\Com(1,r_1),\dots,\com_n=\Com(n,r_n)\)$ of the language $M'$, namely $\ct = 
\Encrypt_{M'}((\com_1,\dots,\com_n), \secret)$. Finally, the share of party 
$\party_i$ is set to be $\langle r_1,\ct \rangle$.

Clearly, if $\Encrypt_{M'}$ and $\Com$ are efficient, then the generation of 
the
shares is efficient. Moreover, the reconstruction procedure is the natural one:
Given a subset of parties $\sop\subseteq \parties$ such that $M(\sop)=1$ and a
valid witness $w$, decrypt $\ct$ using the shares of the parties $\sop$ and
$w$. By the completeness of the witness encryption scheme, given a valid subset
of parties $\sop$ and a valid witness $w$ the decryption will output the secret
$\secret$.

As for the \emph{security} of this scheme, we want to show that it is impossible
to extract (or even learn anything about) the secret having a subset of parties
$\sop$ for which $M(\sop)=0$ (\ie an ``unqualified'' subset of parties). Let 
$\sop$ be such that $M(X)=0$ and let $D$
be an algorithm that extracts the 
secret given the shares of parties corresponding to $\sop$. Roughly speaking, 
we will
use the ability to extract the secret in order to solve the following task: we
are given a list of $n$ unopened string commitments $\com_1,\dots,\com_n$ and a
promise that it either corresponds to the values $A_0=\{1,\dots,n\}$ or it
corresponds to the values $A_1= \{n+1,\dots,2n\}$ and we need to decide which is
the case. Succeeding in this task would break the security guarantee of the
commitment scheme.

We sample $n$ openings $r_1,\dots,r_n$ uniformly at random and create a new
witness encryption $\ct'$ such that $\ct' =
\Encrypt_{M'}((\com'_{1},\dots,\com'_{n}), \secret)$ as above, where we replace
the commitments corresponding to parties not in $\sop$ with commitments from the
input as follows:
\begin{align*}
  \forall i\in[n]:\; \com'_i =
  \begin{cases}
    \Com(i,r_i) & \text{if } \party_i \in \sop \\
    \com_i & \text{otherwise.}
  \end{cases}
\end{align*}
For $i\in[n]$ we set the share of party $\party_i$ to be $\langle r_i,\ct'
\rangle$. We run $D$ with this new set of shares. If we are in the case where
$\com_1,\dots,\com_n$ corresponds to $A_0$, then $D$ is unable to distinguish
between $\ct$ and $\ct'$ and, hence, will be able to extract the secret. On the
other hand, if $\com_1,\dots,\com_n$ corresponds to $A_1$, then there is no
valid witness to decrypt $\ct'$ (since the commitment scheme is
perfectly-binding). Therefore, by the security of the witness encryption scheme,
it is computationally hard to learn anything about the secret $\secret$ from
$\ct'$. Hence, if $D$ is able to extract the secret $\secret$, then we deduce
that $\com_1,\dots,\com_n$ correspond to $A_0$ and, otherwise we conclude that
$\com_1,\dots,\com_n$ correspond to $A_1$.

The above gives intuition for proving security in the non-uniform setting. 
To see this, we assume that there exists an $\sop$ such that $M(\sop)=0$ and 
the 
distinguisher $D$ can extract the secret from the shares of $\sop$. Our 
security definition (see \Cref{sec:RudichSecretSharingDefs}) is uniform and 
requires the distinguisher $D$ to find such 
an $\sop$ \emph{and} extract the secret with noticeable probability. In the 
uniform 
case, we first run $D$ to get $\sop$ and must make sure that $M(\sop)=0$. 
Otherwise, if $M(\sop)=1$, in both cases (that $\com_1,\dots,\com_n$ 
correspond to $A_0$ or to $A_1$) it is easy to extract the secret and thus we 
might be completely fooled. The problem is that $M$ is a language in $\mNP$ 
and, in general, it could be hard to test whether $M(\sop)=0$. We overcome 
this 
by sampling many subsets $\sop$ and use $D$ to estimate which one 
to use. For more information we refer to \Cref{sec:security_proof}.

%%% Local Variables:
%%% TeX-master: "InstanceOptimal"
%%% End:
\section{Preliminaries}
We start with some general notation. We denote by $[n]$ the set of numbers
$\{1,2,\dots,n\}$. Throughout the paper we use $n$ as our security parameter.
We denote by $\U_n$ the uniform distribution on $n$ bits. For a distribution or
random variable $R$ we write $r\leftarrow R$ to denote the operation of sampling
a random element $r$ according to $R$. For a set $S$, we write $s \uniran S$ to
denote the operation of sampling an $s$ uniformly at random from the set $S$. We
denote by $\negl:\N\to\R$ a function such that for every positive integer $c$
there exists an integer $N_c$ such that for all $n > N_c$, $\negl(n) < 1/n^c$.
\subsection{Monotone $\NP$}
A function $f:2^{[n]}\to\B$ is said to be \textsf{monotone} if for every
$X\subseteq [n]$ such that $f(X)=1$ it also holds that $\forall Y\subseteq [n]$
such that $X\subseteq Y$ it holds that $f(Y)=1$.

A \textsf{monotone Boolean circuits} is a Boolean circuit with AND and OR gates
(without negations). A \textsf{non-deterministic circuit} is a Boolean circuit
whose inputs are divided into two parts: standard inputs and non-deterministic
inputs. A non-deterministic circuit accepts a standard input if and only if
there is some setting of the non-deterministic input that causes the circuit to
evaluate to 1. A \textsf{monotone non-deterministic circuit} is a
non-deterministic circuit where the monotonicity requirement applies only to the
standard inputs, that is, every path from a standard input wire to the output
wire does not have a negation gate.

\begin{definition}[\cite{GrigniS90}]\label{def:mNP1}
  We say that a function $L$ is in $\mNP$ if there exists a uniform family of
  polynomial-size monotone non-deterministic circuit that computes $L$.
\end{definition}

\begin{lemma}[{\cite[Theorem 2.2]{GrigniS90}}]\label{def:mNP2}
  $\mNP = \NP \cap \mono$, where $\mono$ is the set of all monotone
  functions.
\end{lemma}

\subsection{Computational Indistinguishability}
\begin{definition}\label{def:computationalIndis}
  Two sequences of random variables $X = \{ X_n\}_{n \in \N}$ and $Y =
  \{Y_n\}_{n \in \N}$ are \textsf{computationally indistinguishable} if for
  every probabilistic polynomial-time algorithm $A$ there exists an integer $N$
  such that for all $n\geq N$,
  \begin{align*}
    \abs{ \Pr[A(X_n) = 1] - \Pr[A(Y_n) = 1] } \leq \negl(n).
  \end{align*}
  where the probabilities are over $X_n$, $Y_n$ and the internal randomness of
  $A$.
\end{definition}

\subsection{Secret-Sharing}
A perfect (\resp computational) secret-sharing scheme involves a dealer who has
a secret, a set of $n$ parties, and a collection $A$ of ``qualified'' subsets of
parties called the access structure. A secret-sharing scheme for $A$ is a method
by which the dealer (\resp efficiently) distributes shares to the parties such
that (1) any subset in $A$ can (\resp efficiently) reconstruct the secret from
its shares, and (2) any subset not in $A$ cannot (\resp efficiently) reveal any
partial information on the secret. For more information on secret-sharing
schemes we refer to \cite{Beimel11} and references therein.

Throughout this paper we deal with secret-sharing schemes for access structures
over
% in $\NP$. Let $M\in \NP$ be a Boolean monotone function on subsets of
$n$ parties $\parties=\parties_n=\{\party_1,\dots,\party_n\}$. % $M$ defines
% an \emph{access structure} on $\parties$.

\begin{definition}[Access structure]
  An \textsf{access structure} $M$ on $\parties$ is a monotone set of subsets of
  $\parties$. That is, for all $\sop\in M$ it holds that $\sop
  \subseteq \parties$ and for all $\sop\in M$ and $\sop'$ such that
  $\sop\subseteq \sop'\subseteq \parties$ it holds that $\sop'\in M$.
\end{definition}
We may think of $M$ as a characteristic function $M:2^{\parties}\to\B$ that
outputs $1$ given as input $\sop\subseteq \parties$ if and only if $\sop$ is in
the access structure.

Many different definitions for secret-sharing schemes appeared in the
literature. Some of the definitions were not stated formally and in some cases
rigorous security proofs were not given. Bellare and Rogaway \cite{BellareR07}
survey many of these different definitions and recast them in the tradition of
provable-security cryptography. They also provide some proofs for well-known
secret-sharing schemes that were previously unanalyzed. We refer to
\cite{BellareR07} for more information.

\subsection{Witness Encryption}
\begin{definition}[Witness encryption \cite{GentryLW14}]\label{def:we}
  A \textsf{witness encryption} scheme for an $\NP$ language $L$ (with a
  corresponding relation $R$) consists of the following two polynomial-time
  algorithms:
\begin{itemize}
\item[] $\Encrypt(1^\lambda, x, M)$: Takes as input a
  security parameter $1^\lambda$, an unbounded-length string $x$ and an message
  $M$ of polynomial length in $\lambda$, and outputs a ciphertext $\ct$.
\item[] $\Decrypt(\ct, w)$: Takes as input a
  ciphertext $\ct$ and an unbounded-length string $w$, and outputs a message $M$
  or the symbol $\bot$.
\end{itemize}
These algorithms satisfy the following two conditions:
\begin{enumerate}
\item \textbf{Completeness (Correctness):} For any security parameter $\lambda$,
  any $M\in\B^{\poly(\lambda)}$ and any $x\in L$ such that $R(x,w)$ holds, we
  have that
  \begin{align*}
    \Pr[\Decrypt (\Encrypt(1^\lambda, x, M), w) = M] = 1.
  \end{align*}
\item \textbf{Soundness (Security):} For any probabilistic polynomial-time
  adversary $A$, there exists a negligible function $\negl(\cdot)$, such that
  for any $x\notin L$ and equal-length messages $M_1$ and $M_2$ we have that
  \begin{align*}
    \abs{\Pr[A(\Encrypt(1^\lambda, x, M_1)=1] - \Pr[A(\Encrypt(1^\lambda, x,
      M_2)=1]} \leq \negl(\lambda).
  \end{align*}
\end{enumerate}
\end{definition}

\paragraph{Remark.} Our definition of Rudich secret-sharing (that is given in
\Cref{sec:RudichSecretSharingDefs}) is uniform. The most common definition of
witness encryption in the literature is a non-uniform one (both in the instance
and in the messages). To achieve our notion of security for Rudich
secret-sharing it is enough to use a witness encryption scheme in which the
messages are chosen uniformly.

\subsection{Commitment Schemes}
In our construction we need a non-interactive commitment scheme such that
commitments of different strings has disjoint support. Since the dealer in the
setup phase of a secret-sharing scheme is not controlled by an adversary (\ie it
is honest), we can relax the foregoing requirement and use non-interactive
commitment schemes that work in the CRS (common random string) model, Moreover,
since the domain of input strings is small (it is of size $2n$) issues of
non-uniformity can be ignored. Thus, we use the following definition:
\begin{definition}[Commitment scheme in the CRS 
model]\label{def:commitment_scheme}
A polynomial-time computable function $\Com \colon \B^{\ell}\times\B^n \times
\B^m \to \B^{*}$, where $\ell$ is the length of the string to commit, $n$ is the
length of the randomness, $m$ is the length of the CRS. We say that $\Com$ is a
(non-interactive perfectly binding) commitment scheme in the CRS model if for
any two inputs $x_1,x_2 \in \B^{\ell}$ such that $x_1 \ne x_2$ it holds that:
  \begin{enumerate}
  \item \emph{Computational Hiding:} Let $\crs\leftarrow \B^m$ be chosen
    uniformly at random. The random variables $\Com(x_1,\U_n,\crs)$ and
    $\Com(x_2,\U_n,\crs)$ are computationally indistinguishable (given
    $\crs$).

  \item \emph{Perfect Binding:} With all but negligible probability over the 
  CRS, the supports of the above random variables are disjoint.
  \end{enumerate}
\end{definition}
Commitment schemes that satisfy the above definition, in the CRS model, can be
constructed based on any pseudorandom generator \cite{Naor91} (which can be
based on any one-way functions \cite{HastadILL99}). For simplicity, throghout
the paper we ignore the CRS % by conditioning on the event that the CRS is
% ``good'' (in the sense that the aforementioned random variables are disjoint)
and simply write $\Com(\cdot,\cdot)$. We say that $\Com(x,r)$ is the
\textsf{commitment} of the value $x$ with the \textsf{opening} $r$.

\section{The Definition of Rudich
  Secret-Sharing} \label{sec:RudichSecretSharingDefs}

In this section we formally define computational secret-sharing for access
structures realizing monotone functions in $\NP$, which we call \emph{Rudich
  secret-sharing}.  Even though secret-sharing schemes for functions in $\NP$
were considered in the past \cite{VinodNSRK03,Beimel11,GargGSW13}, no formal
definition was given.

Our definition consists of two requirements: completeness
and security. The \emph{completeness} requirement assures that a ``qualified''
subset of parties that wishes to reconstruct the secret and \emph{knows} the
witness will be successful. The \emph{security} requirement guarantees that as
long as the parties form an ``unqualified'' subset, they are unable to learn
the secret.

Note that the security requirement stated above is possibly hard to check
efficiently: For some access structures in $\mNP$ (\eg monotone $\NP$-complete
problems) it might be computationally hard to verify that the parties form an
``unqualified'' subset. Next, in \Cref{def:rudichSecretSharing} we give a
\emph{uniform} definition of secret-sharing for $\NP$. In
\Cref{sec:alternative} we give an alternative definition and show their
equivalence.

\begin{definition}[Rudich secret-sharing]\label{def:rudichSecretSharing}
  Let $M:2^\parties\to\B$ be an access structure corresponding to a language $L \in \mNP$
  and let $\ver_M$ be a verifier for $L$. A secret-sharing scheme $\Sc$ for
  $M$ consists of a setup procedure
  $\SETUP$ and a reconstruction procedure $\RECON$ that satisfy the following
  requirements:
  \begin{enumerate}
  \item $\SETUP(1^n,\secret)$ gets as input a secret $\secret$ and distributes a
    share for each party. For $i\in[n]$ denote by $\share(\secret,i)$ the
    random variable that corresponds to the share of party $\party_i$.
    Furthermore, for $\sop\subseteq \parties$ we denote by
    $\shares(\secret, \sop)$ the random variable that corresponds to the set
    of shares of parties in $\sop$.

  \item\label{def:ss_completeness} Completeness:
  	
    If $\RECON(1^n,\shares(\secret,\sop),w)$ gets as input the shares of a
    ``qualified'' subset of parties and a valid witness, and outputs the shared
    secret. Namely, for $\sop\subseteq
    \parties$ if $M(\sop) = 1$, then for any valid witness $w$ such that
    $\ver_M(\sop,w)=1$, it holds that:
    \begin{align*}
      \Pr\left[\RECON(1^n,\shares(\secret, \sop), w) =
        \secret\right] = 1,
    \end{align*}
    where the probability is over the internal randomness of the scheme and of
    $\RECON$.

  \item\label{def:ss_indist} Indistinguishability of the Secret:

    For every pair of probabilistic polynomial-time algorithms $(\Samp,D)$ where
    $\Samp(1^n)$ defines a distribution over pairs of secrets
    $\secret_0,\secret_1$, a subset of parties $\sop$ and auxiliary information
    $\sigma$, it holds that
    \begin{align*}
      | \Pr&\left[M(\sop) = 0 \myand D(1^n, \secret_0,\secret_1,
        \shares(\secret_0,\sop),\sigma)
        = 1 \right]  - \\
      & \Pr\left[M(\sop) = 0 \myand D(1^n, \secret_0,\secret_1,
        \shares(\secret_1,\sop),\sigma) = 1 \right] | \leq \negl(n),
    \end{align*}
    where the probability is over the internal randomness of the scheme, the
    internal randomness of $D$ and the distribution
    $(\secret_0,\secret_1,\sop,\sigma)\leftarrow \Samp(1^n)$.

    That is, for every pair of probabilistic polynomial-time algorithms
    $(\Samp,D)$ such that $\Samp$ chooses two secrets $\secret_0,\secret_1$ and
    a subset of parties $\sop\subseteq \parties$, if $M(\sop)=0$ then $D$ is
    unable to distinguish (with noticeable probability) between the shares of
    $X$ generated by $\SETUP(\secret_0)$ and the shares of $X$ generated by
    $\SETUP(\secret_1)$.
  \end{enumerate}
\end{definition}

\paragraph{Notation.} For ease of notation, $1^n$ and $\sigma$ are omitted
when they are clear from the context.

\subsection{An Alternative Definition: Semantic Security}\label{sec:alternative}
The security requirement (\ie the third requirement) of a Rudich secret-sharing
scheme that is given in \Cref{def:rudichSecretSharing} is phrased in the spirit
of \emph{computational indistinguishability}.  A different approach is to define
the security of a Rudich secret-sharing in the spirit of \emph{semantic
  security}. As in many cases (\eg encryption \cite{GoldwasserM84}), it turns
out that the two definitions are equivalent.

\begin{definition}[Rudich secret-sharing - semantic security
  version]\label{def:rudichSecretSharingSemanticSecurity}

  Let $M:2^\parties\to\B$ be an $\mNP$ access structure with verifier
  $\ver_M$. A secret-sharing scheme $\Sc$ for $M$ consists of a setup procedure
  $\SETUP$ and a reconstruction procedure $\RECON$ as in
  \Cref{def:rudichSecretSharing} and has the following property \emph{instead}
  of the \emph{indistinguishability of the secret} property:
  \begin{enumerate}
  \item[3]\label{def:ss_unlearn} Unlearnability of the Secret:

    For every pair of probabilistic polynomial-time algorithms $(\Samp,D)$ where
    $\Samp(1^n)$ defines a distribution over a secret $S$, a subset of parties
    $\sop$ and auxiliary information $\sigma$, and for every efficiently
    computable function $f:\B^*\to\B^*$ it holds that there exists a
    probabilistic polynomial-time algorithm $D'$ (called a \emph{simulator})
    such that
    \begin{align*}
      | \Pr&\left[M(\sop) = 0 \myand D(1^n, \shares(\secret,\sop), \sigma)
        = f(\secret) \right]  - \\
      & \Pr\left[M(\sop) = 0 \myand D'(1^n, \sop, \sigma)= f(\secret) \right] |
      \leq \negl(n),
    \end{align*}
    where the probability is over the internal randomness of the scheme, the
    internal randomness of $D$ and $D'$, and the distribution $(\secret, X,
    \sigma)\leftarrow \Samp(1^n)$.

    That is, for every pair of probabilistic polynomial-time algorithms
    $(\Samp,D)$ such that $\Samp$ chooses a secret $\secret$ and a subset of
    parties $\sop\subseteq \parties$, if $M(\sop)=0$ then $D$ is unable to learn
    anything about $\secret$ that it could not learn without access to the
    secret shares of $\sop$.
  \end{enumerate}
\end{definition}

\def\definitionEquiv{\Cref{def:rudichSecretSharingSemanticSecurity} and
  \Cref{def:rudichSecretSharing} are equivalent.}
\begin{theorem}\label{thm:def_equiv}
\definitionEquiv
\end{theorem}
We defer the proof of \Cref{thm:def_equiv} to \Cref{sec:def_equiv}.

\subsection{Definition of Adaptive Security}\label{sec:adaptive_ss}
Our definition of Rudich secret-sharing only guarantees security against static
adversaries. That is, the adversary chooses a subset of parties before it sees
any of the shares. In other words, the selection is done \emph{independently} of
the sharing process and hence, we may think of it as if the sharing process is
done \emph{after} $\Samp$ chooses $\sop$.

A stronger security guarantee would be to require that even an adversary that
chooses its set of parties in an \emph{adaptive} manner based on the shares it
has seen so far is unable to learn the secret (or any partial information about
it). Namely, the adversary chooses the parties one by one depending on the
secret shares of the previously chosen parties.

% Not entirely clear that this is where we want this here or later, but
%certainly after the formal definition

The security proof of our scheme (which is given in \Cref{sec:ss_from_io}) does
not hold under this stronger requirement. It would be interesting to strengthen
it to the adaptive case as well. One problem that immediately arises in an
analysis of our scheme against adaptive adversaries is that of \emph{selective
  decommitment} (cf.\ \cite{DworkNRS03}), that is when an adversary sees a
collection of commitments and can select a subset of them and receive their
openings. The usual proofs of security of commitment schemes are not known to
hold in this case.

\section{Rudich Secret-Sharing from Witness Encryption}\label{sec:ss_from_io}
In this section we prove the main theorem of this paper. We show how to
construct a Rudich secret-sharing scheme for any $\mNP$ access structure
assuming witness encryption for $\NP$ and one-way functions.

\begin{reptheorem}{thm:main}[Restated]
  Assuming witness encryption for $\NP$ and one-way functions, there is an
  efficient computational secret-sharing scheme for every $\mNP$ access
  structure.
\end{reptheorem}

Let $\parties = \{\party_1,\dots,\party_n\}$ be a set of $n$ parties and let
$M:2^{\parties}\to\B$ be an $\mNP$ access structure. We view $M$ either as a
function or as a language. For a language $L$ in $\NP$ let
$(\Encrypt_L,\Decrypt_L)$ be a witness encryption scheme and let
$\Com:[2n]\times\B^{n}\to\B^{q(n)}$ be a commitment scheme, where $q(\cdot)$ is
a polynomial.

\paragraph{The Scheme.} 
We define a language $M'$ that is related to $M$ as follows. The language $M'$
consists of sets of $n$ strings $\{\com_i\}_{i\in[n]}\in\B^{q(n)}$ as follows. 
$M'(\com_1,\dots,\com_n) = 1$
if and only if there exist $\{r_i\}_{i\in[n]}$ such that
$M(x) = 1$, where $x \in \B^n$ is such that
\begin{align*}
  \forall i\in[n]:\; \quad x_i =
  \begin{cases}
    1 & \text{if } r_i \neq \bot \text{ and }\Com(i, r_i) =
    \com_i,\\
    0 & \text{otherwise.}
  \end{cases}
\end{align*}

For every $i\in[n]$, the \emph{share} of party $\party_i$ is composed of 2
components: (1) $r_i\in\B^{n}$ - an opening of a commitment to the value $i$,
and (2) a witness encryption $\ct$. The witness encryption encrypts the secret
$\secret$ with respect to the commitments of all parties $\{\com_i =
\Com(i,r_i)\}_{i\in[n]}$. To reconstruct the secret given a subset of parties
$\sop$, we simply decrypt $\ct$ given the corresponding openings of $\sop$ and
the witness $w$ that indeed $M(\sop)=1$. The secret-sharing scheme is formally
described in \Cref{fig:SS}.
\begin{figure}[h]
  \begin{boxedminipage}{\textwidth}
    \small \medskip \noindent

    \vspace{3mm} \textbf{The Rudich Secret-Sharing Scheme $\Sc$ for $M$}

    \vspace{4mm} \textbf{The $\SETUP$ Procedure:}

    \vspace{3mm} \emph{Input}: A secret $\secret$.

    \vspace{1mm} Let $M'$ be the language as described above, and let 
    $(\Encrypt_{M'},\Decrypt_{M'})$ be a witness encryption for $M'$ (see 
    \Cref{def:we}).

    \begin{enumerate}
    \item\label{fig_ss:1}\text{For $i\in[n]$:}
      \begin{enumerate}
      %\item Sample uniformly at random an identifier $a_i\in\B^n$.
      \item Sample uniformly at random an opening $r_i\in\B^n$.
      %\item Compute the commitment $\com_i = \Com( \langle i, a_i \rangle,
      %  r_i)$.
      \item Compute the commitment $\com_i = \Com( i, r_i)$.
      \end{enumerate}

    \item Compute $\ct \leftarrow \Encrypt_{M'}((\com_1,\dots,\com_n), 
    \secret)$.

    \item Set the share of party $\party_i$ to be $\share(\secret,i)= \langle
    r_i, \ct \rangle$.

    \end{enumerate}

    \vspace{1mm}

    \vspace{3mm} \textbf{The $\RECON$ Procedure:}

    \vspace{3mm} \emph{Input}: A non-empty subset of parties $
    \sop\subseteq\parties$ together with their shares and a witness $w$ of
    $\sop$ for $M$.

    \begin{enumerate}
    \item Let $\ct$ be the witness encryption in the shares of $\sop$.
    \item For any $i \in [n]$ let
	$ 
	  r'_i =
      \begin{cases}
        r_i & \text{if } p_i \in X  \\
        \bot & \text{otherwise.}
      \end{cases}
    $
    \item Output $\Decrypt_{M'}(\ct, (r'_1,\dots,r'_n, w))$.
    \end{enumerate}

    \vspace{1mm}

  \end{boxedminipage}

  \caption{Rudich secret-sharing scheme for $\NP$.}
  \label{fig:SS}
\end{figure}

Observe that if the witness encryption scheme and $\Com$ are both efficient, 
then
the scheme is efficient (\ie $\SETUP$ and $\RECON$ are probabilistic
polynomial-time algorithms). $\SETUP$ generates $n$ commitments and a witness
encryption of polynomial size. $\RECON$ only decrypts this witness encryption.

\paragraph{Completeness.} The next lemma states that the scheme is
complete. That is, whenever the scheme is given a qualified
$\sop\subseteq \parties$ and a valid witness $w$ of $\sop$, it is possible to
successfully reconstruct the secret.
\begin{lemma}\label{lemma:completeness}
  Let $M \in \NP$ be an $\mNP$ access structure. Let $\Sc = \Sc_M$ be the scheme
  from \Cref{fig:SS} instantiated with $M$. For every subset of parties
  $\sop\subseteq \parties$ such that $M(\sop)=1$ and any valid witness $w$ it
  holds that
  \begin{align*}
    \Prp{\RECON(\shares(\secret, \sop), w) = \secret}=1.
  \end{align*}
\end{lemma}
\begin{proof}
  Recall the definition of the algorithm $\RECON$ from \Cref{fig:SS}: $\RECON$
  gets as input the shares of a subset of parties $\sop =
  \{\party_{i_1},\dots,\party_{i_k}\}$ for $k,i_1,\dots,i_k\in[n]$ and a valid
  witness $w$. Recall that the shares of the parties in $\sop$ consist of $k$
  openings for the corresponding commitments and a witness encryption
  $\ct$. $\RECON$ decrypts $\ct$ given the openings of parties in $\sop$ and 
  the
  witness $w$.

  By the completeness of the witness encryption scheme, the output of the
  decryption procedure on $\ct$, given a valid $\sop$ and a valid witness, is
  $\secret$ (with probability 1).
\end{proof}

\paragraph{Indistinguishability of the Secret.} We show that our scheme is
secure. More precisely, we show that given an ``unqualified'' set of parties
$\sop \subseteq \parties$ as input (\ie $M(\sop)=0$), with overwhelming
probability, any probabilistic polynomial-time algorithm cannot distinguish the
shared secret from another.

To this end, we assume towards a contradiction that
such an algorithm exists and use it to efficiently solve the following task:
given two lists of $n$ commitments and a promise that one of them corresponds to
the values $\{1,\dots,n\}$ and the other corresponds to the values
$\{n+1,\dots,2n\}$, identify which one corresponds to the values
$\{1,\dots,n\}$. The following lemma shows that solving this task efficiently
can be used to break the hiding property of the commitment scheme.

\def\ManytoOne{
  Let $\Com\colon[2n]\times\B^{n}\to\B^{q(n)}$ be a commitment scheme
  where
  $q(\cdot)$ is a polynomial. If there exist $\eps = \eps(n) > 0$ and a
  probabilistic polynomial-time algorithm $D$ for which
  \begin{align*}
    | \Pr&[D(\Com(1,\U_{n}),\dots,\Com(n,\U_{n}))=1] - \\
    &\Pr[D(\Com(n,\U_{n}),\dots, \Com(2n,\U_{n}))=1] | \ge \eps,
  \end{align*}
  then there exist a probabilistic polynomial-time algorithm $D'$ and
  $x,y\in[2n]$ such that
  \begin{align*}
    \abs{ \Pr[D'(\Com(x,\U_{n}))=1] - \Pr[D'(\Com(y,\U_{n}))=1] } \ge
    \eps/n.
  \end{align*}
}
\begin{lemma}\label{lemma:list_of_hardness}
\ManytoOne
\end{lemma}
The proof of the lemma follows from a standard hybrid argument. See
full details in \Cref{sec:prf_lemma_lest}.

At this point we are ready to prove the security of our scheme. That is, we show
that the ability to break the security of our scheme translates to the ability
to break the commitment scheme (using \Cref{lemma:list_of_hardness}).

\begin{lemma}\label{lemma:security}
  Let $\parties = \{\party_1,\dots,\party_n\}$ be a set of $n$ parties. Let
  $M:2^\parties\to\B$ be an $\mNP$ access structure. If there exist a
  non-negligible $\eps=\eps(n)$ and a pair of probabilistic polynomial-time
  algorithms $(\Samp,D)$ such that for $(\secret_0,\secret_1,\sop)\leftarrow
  \Samp(1^n)$ it holds that
  \begin{align*}
    \Pr&\left[M(\sop) = 0 \myand D( \secret_0, \secret_1,
      \shares(\secret_0,\sop)) = 1 \right] \\
    & - \Pr\left[M(\sop) = 0 \myand D( \secret_0, \secret_1,
      \shares(\secret_1,\sop)) = 1 \right] \geq \eps,
  \end{align*}
  then there exists a probabilistic algorithm $D'$ that runs in polynomial-time
  in $n/\eps$ such that for sufficiently large $n$
  \begin{align*}
    | \Pr&[D'(\Com(1,\U_{n}),\dots, \Com(n,\U_{n})
      )=1] - \\
    & \Pr[D'(\Com(n+1,\U_{n}),\dots,\Com(2n,\U_{n}) )=1] | \ge \eps/10
    - \negl(n).
    \end{align*}
\end{lemma}

The proof of \Cref{lemma:security} appears in \Cref{sec:security_proof}.

Using \Cref{lemma:security} we can prove \Cref{thm:main}, the main theorem of
this section. The \emph{completeness} requirement (\Cref{def:ss_completeness} in
\Cref{def:rudichSecretSharing}) follows directly from
\Cref{lemma:completeness}. The \emph{indistinguishability of the secret}
requirement (\Cref{def:ss_indist} in \Cref{def:rudichSecretSharing}) follows by
combining \Cref{lemma:security,lemma:list_of_hardness} together with the hiding
property of the commitment scheme. \Cref{sec:security_proof} is devoted to the
proof of \Cref{lemma:security}.

\subsection{Main Proof of Security}\label{sec:security_proof}
Let $M$ be an $\mNP$ access structure, $(\Samp,D)$ be a pair of algorithms and
$\eps>0$ be a function of $n$, as in the \Cref{lemma:security}. We are
given a
list of
(unopened) string commitments
$\com_1,\dots,\com_n\in\{\Com(z_i,r)\}_{r\in\B^n}$, where for $Z =
\{z_1,\dots,z_n\}$ either $Z=\{1,\dots,n\} \triangleq A_0$ or
$Z=\{n+1,\dots,2n\} \triangleq A_1$. Our goal is to construct an algorithm $D'$
that distinguishes between the two cases (using $\Samp$ and $D$) with
non-negligible probability (that is related to $\eps$). Recall that $\Samp$
chooses two secrets $\secret_0, \secret_1$ and $\sop\subseteq \parties$ and then
$D$ gets as input the secret shares of parties in $\sop$ for one of the
secrets. By assumption, for $(\secret_0,\secret_1,\sop)\leftarrow
\Samp(1^n)$ we
have that
\begin{align}\label{eq:assumption}
  |\Pr&\left[M(\sop)=0 \myand D( \secret_0, \secret_1,
    \shares(\secret_0,\sop)) = 1 \right] - \nonumber
  \\
  & \Pr\left[M(\sop)=0 \myand D( \secret_0, \secret_1,
    \shares(\secret_1,\sop)) = 1 \right]| \geq \eps.
\end{align}
% We assume, without loss of generality, that
% \begin{align}\label{eq:assumption}
%   \Pr&\left[M(\sop)=0 \myand D(\shares(\secret_1,\sop)) = 1 \right] - \nonumber
%   \\
%   & \Pr\left[M(\sop)=0 \myand D(\shares(\secret_0,\sop)) = 1 \right] \geq
%   \eps.
% \end{align}

Roughly speaking, the algorithm $D'$ that we define creates a new set of shares
using $\com_1,\dots,\com_n$ such that: If $\com_1,\dots,\com_n$ are commitments
to $Z=A_0$ then $D$ is able to recover the secret; otherwise, (if $Z=A_1$) it is
computationally hard to recover the secret. Thus, $D'$ can distinguish between
the two cases by running $D$ on the new set of shares and acting according to
its output.

We begin by describing a useful subroutine we call $\Dver$. The inputs to
$\Dver$ are $n$ string commitments $\com_1,\dots,\com_n$, two secrets
$\secret_0,\secret_1$ and a subset of $k\in[n]$ parties $\sop$. Assume for ease
of notations that $\sop = \{\party_1,\dots,\party_k\}$. $\Dver$ first chooses
$b$ uniformly at random from the set $\B$ and samples uniformly at random $n$
openings $r_1,\dots,r_n$ from the distribution $\U_n$. Then, $\Dver$ computes
the witness encryption $\ct'_b$ of the message $\secret_b$ with respect to the
instance $\Com(1,r_1),\dots,\Com(k,r_k),\com_{k+1},\dots,\com_{n}$ of $M'$ (see
\Cref{fig:SS}) and sets for every $i\in[n]$ the share of party $\party_i$ to be
$\shares_{}'(\secret_b,i)= \langle r_i, \ct'_b \rangle$. Finally, $\Dver$
emulates the execution of $D$ on the set of shares of $\sop$
($\shares_{}'(\secret_b,\sop)$). If the output of $D$ equals to $b$, then
$\Dver$ outputs $1$ (meaning the input commitments correspond to $Z=A_0$);
otherwise, $\Dver$ outputs $0$ (meaning the input commitments correspond to
$Z=A_1$).

The na\"{i}ve implementation of $D'$ is to run $\Samp$ to generate
$\secret_0,\secret_1$ and $\sop$, run $\Dver$ with the given string
commitments,
$\secret_0,\secret_1$ and $\sop$, and output accordingly. This, however, does
not work. To see this, recall that the assumption (\cref{eq:assumption}) only
guarantees that $D$ is able to distinguish between the two secrets when
$M(\sop)=0$. However, it is possible that with high probability (yet smaller
than
$1-1/\mathsf{poly}(n)$) over $\Samp$ it holds that $M(\sop)=1$, in which we do
not have any guarantee on $D$. Hence, simply running $\Samp$ and $\Dver$ might
fool us in outputting the wrong answer.

The first step to solve this is to observe that, by the assumption in
\cref{eq:assumption}, $\Samp$ generates an $\sop$ such that $M(\sop)=0$ with
(non-negligible) probability at least $\eps$. By this observation, notice that
by running $\Samp$ for $\Theta(n/\eps)$ iterations we are assured that with very
high probability (specifically, $1-\negl(n)$) there exists an iteration in which
$M(\sop)=0$. All we are left to do is to recognize in which iteration
$M(\sop)=0$ and only in that iteration we run $\Dver$ and output accordingly.

However, in general it might be computationally difficult to test for a given
$\sop$ whether $M(\sop)=0$ or not. To overcome this, we observe that we need
something much simpler than testing if $M(\sop)=0$ or not. All we actually need
is a procedure that we call $\Mest$ that checks if $\Dver$ is a good
distinguisher (between commitments to $A_0$ and commitments to $A_1$) for a
given $\sop$. On the one hand, by the assumption, we are assured that this is
indeed the case if $M(X)=0$. On the other hand, if $M(X)=1$ and $\Dver$ is
biased, then simply running $\Dver$ and outputting accordingly is enough.

Thus,
our goal is to estimate the bias of $\Dver$. The latter is implemented
efficiently by running $\Dver$ independently $\Theta(n/\eps)$ times on both
inputs (\ie with $Z=A_0$ and with $Z=A_1$) and counting the number of
``correct'' answers.

Recapping, our construction of $D'$ is as follows: $D'$ runs for
$\Theta(n/\eps)$ iterations such that in each iteration it runs $\Samp(1^n)$ and
gets two secrets $\secret_0,\secret_1$ and a subset of parties $\sop$. Then, it
estimates the \emph{bias} of $\Dver$ for that specific $\sop$ (independently of
the input). If the bias is large enough, $D'$ evaluates $\Dver$ with the input
of $D'$, the two secrets $\secret_0,\secret_1$ and the subset of parties $\sop$
and outputs its output. The formal description of $D'$ is given in
\Cref{fig:D'}.

\begin{figure}[!h]
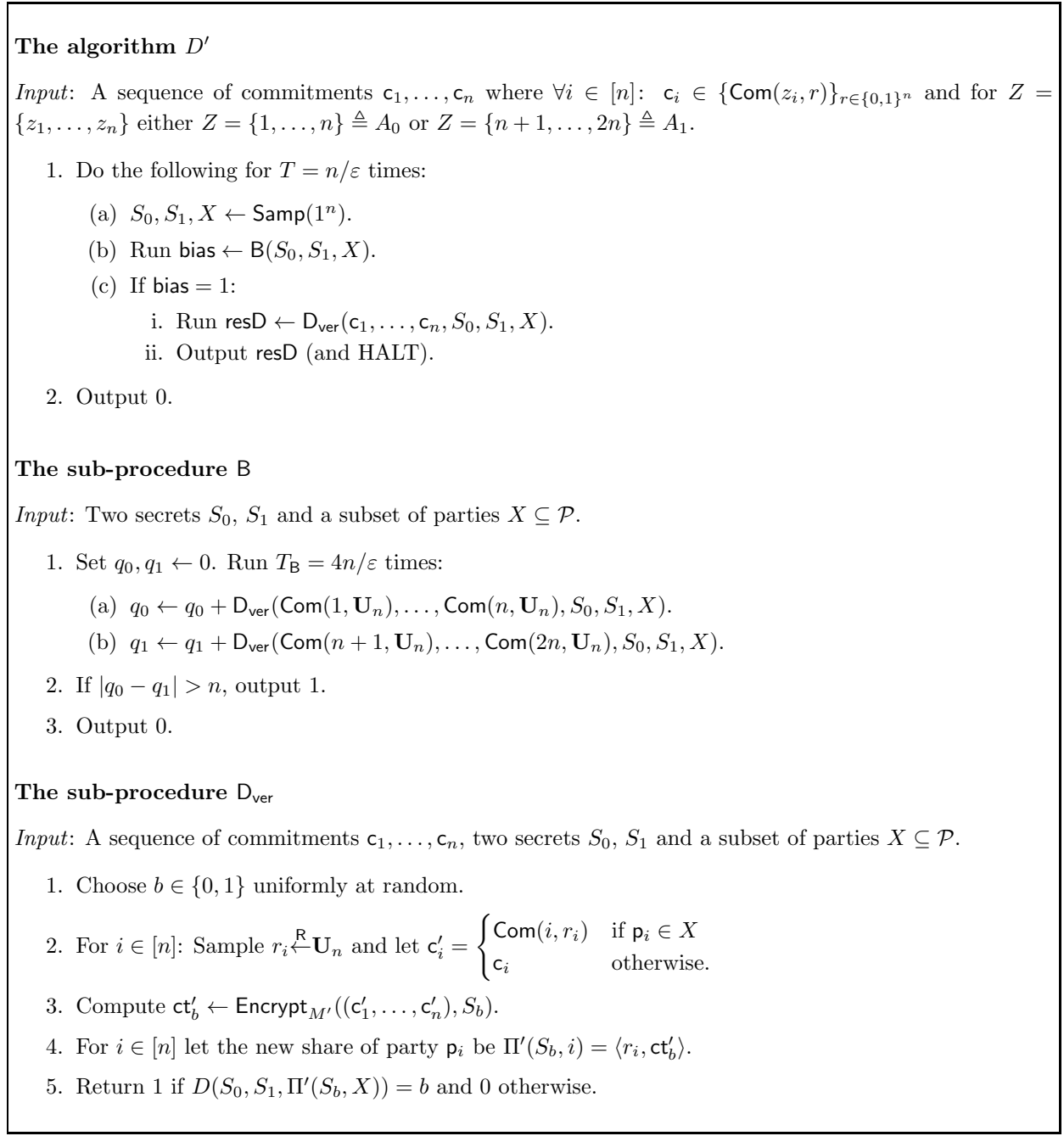

  \begin{boxedminipage}{\textwidth}
    \small \medskip \noindent

    \vspace{2mm} \textbf{The algorithm $D'$}

    \vspace{3mm} \emph{Input}: A sequence of commitments $\com_1, \dots,
    \com_n$ where $\forall
    i\in[n]\colon\;\com_i\in\{\Com(z_i,r)\}_{r\in\B^n}$ and for $Z =
    \{z_1,\dots,z_n\}$ either $Z=\{1,\dots,n\} \triangleq A_0$ or
    $Z=\{n+1,\dots,2n\} \triangleq A_1$.

    %\vspace{3mm} \emph{Algorithm:}
    \begin{enumerate}
    \item Do the following for $T=n/\eps$ times:
      \begin{enumerate}
      \item $\secret_0, \secret_1, \sop \leftarrow \Samp(1^n)$.
      \item Run $\resm \leftarrow \Mest(\secret_0,\secret_1,\sop)$.
      \item If $\resm = 1$:
        \begin{enumerate}
        \item Run $\resd \leftarrow \Dver(\com_1,\dots,\com_n,
          \secret_0,\secret_1,\sop)$.
        \item Output $\resd$ (and HALT).
        \end{enumerate}

      \end{enumerate}
    \item\label{item:default} Output 0.
    \end{enumerate}
    \vspace{1mm}

    \vspace{3mm} \textbf{The sub-procedure $\Mest$}

    \vspace{3mm} \emph{Input}: Two secrets $\secret_0$, $\secret_1$ and a
    subset of parties $\sop\subseteq \parties$.

    \begin{enumerate}
    \item Set $q_0,q_1\leftarrow 0$. Run $T_{\Mest}=4n/\eps$ times:
      \begin{enumerate}
      \item $q_0 \leftarrow q_0 + \Dver(\Com(1,\U_n),\dots,\Com(n,\U_n),
        \secret_0,\secret_1,\sop)$.
      \item $q_1 \leftarrow q_1 + \Dver(\Com(n+1,\U_n),\dots,\Com(2n,\U_n),
        \secret_0,\secret_1,\sop)$.
      \end{enumerate}
    \item If $|q_0-q_1| > n$, output 1.
    \item Output 0.
    \end{enumerate}

    \vspace{3mm} \textbf{The sub-procedure $\Dver$}

    \vspace{3mm} \emph{Input}: A sequence of commitments $\com_1, \dots,
    \com_n$, two secrets $\secret_0$, $\secret_1$ and a subset of parties
    $\sop\subseteq \parties$.

    \begin{enumerate}
    \item Choose $b\in\B$ uniformly at random.
    \item For $i\in[n]$: Sample $r_i\uniran\U_n$ and let $\com'_i =
      \begin{cases}
        \Com(i,r_i) & \text{if } \party_i \in \sop \\
        \com_i & \text{otherwise.}
      \end{cases}$
    \item Compute $\ct_b' \leftarrow \Encrypt_{M'}((\com_1',\dots,\com_n'), 
    \secret_b)$.
    \item For $i\in[n]$ let the new share of party $\party_i$ be
      $\shares'_{}(\secret_b,i) = \langle r_i,\ct_b' \rangle$.
        % \item Run $D$ on $\shares'_b(\secret_b,\sop)$.
      \item Return $1$ if $D( \secret_0, \secret_1,
        \shares'_{}(\secret_b,\sop))=b$ and $0$ otherwise.
      \end{enumerate}
      \vspace{1mm}

    \end{boxedminipage}

  \caption{The description of the algorithm $D'$.}
  \label{fig:D'}
\end{figure}

\paragraph{Analysis of $D'$. }
We prove the following lemma which is a restatement of \Cref{lemma:security}.
\begin{replemma}{lemma:security}[Restated]
  Let $\com_1,\dots,\com_n\in\{\Com(z_i,r)\}_{r\in\B^n}$ be a list of string
  commitments, where for $Z = \{z_1,\dots,z_n\}$ either $Z=\{1,\dots,n\}
  \triangleq A_0$ or $Z=\{n+1,\dots,2n\} \triangleq A_1$. Assuming
  \cref{eq:assumption}, it holds that
  \begin{align*}
    |\Pr[D'(\com_1,\dots,\com_n) = 1 \;|\; Z=A_0] - \Pr[D'(\com_1,\dots,\com_n) =
    1 \;|\; Z=A_1] | \geq \eps/10 -\negl(n).
  \end{align*}
\end{replemma}

We begin with the analysis of the procedure $\Dver$. In the next two claims we
show that assuming that $M(\sop)=0$, then $\Dver$ is a good distinguisher
between the case $Z=A_0$ and the case $Z=A_1$. Specifically, the first claim
states that $\Dver$ answers correctly given input $Z=A_0$ with probability at
least $1/2 + \eps/2$ while in the second claim we show that $\Dver$ is unable to
do much better than merely guessing given input $Z=A_1$ (assuming $M(\sop)=0$).

\begin{claim}\label{claim:mainProofClaim2}
  For $(\secret_0,\secret_1,\sop)\leftarrow \Samp(1^n)$ it holds that
  \begin{align*}
    | \Pr\left[ \Dver(\com_1,\dots,\com_n,\secret_0,\secret_1,\sop)=1
    \;|\;
      M(\sop)=0 \myand Z=A_0 \right] -1/2 | \geq \eps/2.
  \end{align*}
\end{claim}
\begin{proof}
  By the definition of $\Dver$ (see \Cref{fig:D'}) we have that
  $\Dver(\com_1,\dots,\com_k,\secret_0,\secret_1,\sop)=1$ if and only
  if $D(
  \secret_0, \secret_1, \shares'(\secret_b,\sop)) = b$ for
  $b\uniran\B$. Since
  $b$ is chosen uniformly at random from $\B$, it is enough to show that
  \begin{align*}
    \eps \leq & | \Pr\left[D( \secret_0, \secret_1,
      \shares'(\secret_1,\sop)) = 1 \;|\; M(\sop)=0\right] \nonumber
    \\
    & - \Pr\left[D( \secret_0, \secret_1, \shares'(\secret_0,\sop)) = 1
      \;|\; M(\sop)=0\right] | .
  \end{align*}

  Using the assumption (see \cref{eq:assumption}), for
  $(\secret_0,\secret_1,X)\leftarrow \Samp(1^n)$ it holds that
  \begin{align*}
    \eps \leq& | \Pr\left[M(\sop)=0 \myand D( \secret_0, \secret_1,
      \shares(\secret_1,\sop)) = 1 \right] \nonumber
    \\
    & -\Pr\left[M(\sop)=0 \myand D( \secret_0, \secret_1,
      \shares(\secret_0,\sop)) = 1 \right] | \\
    %= & |\Pr\left[D( \secret_0, \secret_1, \shares(\secret_1,\sop),
    %\sigma) = 1
    %  \;|\; M(\sop)=0\right]\cdot\Pr[M(\sop)=0] \nonumber
    %\\
    %& - \Pr\left[D( \secret_0, \secret_1, \shares(\secret_0,\sop))
    %= 1 \;|\;
    %  M(\sop)=0\right]\cdot\Pr[M(\sop)=0] | \\
    \leq & | \Pr\left[D( \secret_0, \secret_1, \shares(\secret_1,\sop)) = 1
      \;|\; M(\sop)=0\right] \nonumber
    \\
    & - \Pr\left[D( \secret_0, \secret_1, \shares(\secret_0,\sop))
    = 1 \;|\;
      M(\sop)=0\right] |.
  \end{align*}
  Notice that since $Z=A_0$ we have that the sequence
  $(\Com(1,\U_{n}),\dots,\Com(n,\U_n))$
  is \emph{identically} distributed as the sequence $\(\com'_1,\dots,
  \com'_n\)$. Hence, for any $b\in\B$ it holds that $\shares'_{}(\secret_b,
  \sop)$ is identically distributed as $\shares_{}(\secret_b, \sop)$. Hence,
  \begin{align*}
    \eps \leq & | \Pr\left[D( \secret_0, \secret_1,
      \shares'(\secret_1,\sop)) = 1 \;|\; M(\sop)=0\right] \nonumber
    \\
    & - \Pr\left[D( \secret_0, \secret_1, \shares'(\secret_0,\sop)) = 1
      \;|\; M(\sop)=0\right] | ,
  \end{align*}
  as required.
\end{proof}

\begin{claim}\label{claim:mainProofClaim3}
  For $(\secret_0,\secret_1,\sop)\leftarrow \Samp(1^n)$ it holds that
  \begin{align*}
    | \Pr\left[ \Dver(\com_1,\dots,\com_n,\secret_0,\secret_1,\sop)=1
    \;|\;
      M(\sop)=0 \myand Z=A_1 \right] -1/2 | \leq \negl(n).
  \end{align*}
\end{claim}
\begin{proof}
  Recall that $\Dver(\com_1,\dots,\com_n,\secret_0,\secret_1,\sop)=1$
  if and
  only if for $b$ chosen uniformly at random from $\B$ it holds that $D(
  \secret_0, \secret_1, \shares'(\secret_b,\sop)) = b$.

  Recall that for $b\in\B$ and $i\in[n]$ the new share of party $\party_i$
  denoted by $\shares'_{}(\secret_b,i)$ consists of the pair $\langle
  r_i^b,\ct_b') \rangle$ where $r_i^b$ is chosen uniformly at random from
  $\U_n$. To prove the claim we show that $\ct_0'$ and $\ct_1'$ are
  computationally indistinguishable. 

  To this end, we show that if $Z=A_1$ and $M(\sop)=0$, then there is \emph{no}
  witness attesting to the fact that $\com_1',\dots,\com_n'$ is in $M'$. Fix
  $\sop\subseteq\parties$ such that $M(\sop) = 0$ and let $(\{r_i'\}_{i\in[n]},
  w)\in (\B^{n})^n\times\B^*$ be a possible witness. Let $\sop'$ be the set of
  parties that correspond to the $r_i'$'s for which $r_i' \neq \bot$.

  If $\sop'\not\subseteq \sop$, then there exists an $i\in[n] $ such that
  $\party_i\in\sop'$ and $\party_i\notin \sop$. In this case, the witness is
  invalid since for every $i$ such that $\party_i\notin \sop$ the commitment
  $\com_i$ is a commitment to the value $n+i$ (and not $i$). Recall that the
  distributions $\Com(i,\U_n)$ and $\Com(j,\U_n)$ are \emph{disjoint} for every
  $i\neq j$. Hence, any opening for the commitment $\com_i$ and the value $i$ is
  \emph{invalid}, \ie any opening $r'_i$ will fail the test $\com_i
  \stackrel{?}{=} \Com(i,r'_i)$.

  Otherwise, if $\sop' \subseteq \sop$, then since $M$ is monotone and
  $M(\sop)=0$ it holds that $M(\sop')=0$. Therefore, the witness is invalid for
  $\sop'$.

  In conclusion, since $M'(\com_1,\dots,\com_n)=0$, the witness encryptions of
  $\secret_0$ and $\secret_1$ are computationally indistinguishable from one
  another (see \Cref{def:we}) and the claim follows.
\end{proof}

Next, we continue with two claims connecting $\Dver$ and $\Mest$. Before we
state these claims, we introduce a useful notation regarding the bias of the
procedure $\Dver$. We denote by $\bias(\secret_0,\secret_1,\sop)$ the advantage
of $\Dver$ in recognizing the case $Z=A_0$ over the case $Z=A_1$ given two
secrets $\secret_0$ and $\secret_1$ and a subset of parties $\sop$. Namely, for
any $\secret_0,\secret_1$ and $X$ denote
\begin{align*}
  \bias(\secret_0,\secret_1,\sop) =
  | \Pr&\left[\Dver(\Com(1,\U_n),\dots,\Com(n,\U_n),
    \secret_0,\secret_1,\sop)=1\right] \\
  & -\Pr\left[\Dver(\Com(n+1,\U_n),\dots,\Com(2n,\U_n),
    \secret_0,\secret_1,\sop)=1\right] |.
\end{align*}

The first claim states that if $\Dver$ is biased (in the sense that
$\bias(\secret_0,\secret_1,\sop)$ is large enough), then $\Mest$ almost surely
notices that and outputs $1$, and vice-versa, \ie if $\Dver$ is unbiased (in the
sense that $\bias(\secret_0,\secret_1,\sop)$ is small enough), then $\Mest$
almost surely notices that and outputs $0$.

\begin{claim}\label{claim:bias_to_B}
  For $(\secret_0,\secret_1,X)\leftarrow \Samp(1^n)$,
  \begin{enumerate}
  \item $ \Pr[\Mest(\secret_0,\secret_1,\sop) = 1 \;|\;
    \bias(\secret_0,\secret_1,\sop) \geq \eps/3] \geq 1-\negl(n)$
  \item $ \Pr[\Mest(\secret_0,\secret_1,\sop) = 1 \;|\;
    \bias(\secret_0,\secret_1,\sop) \leq \eps /10] \leq \negl(n)$
  \end{enumerate}
\end{claim}
\begin{proof}
  Recall that $\Mest$ runs for $T_{\Mest}$ \emph{independent} iterations such
  that in each iteration it executes $\Dver$ twice: Once with
  $\Com(1,\U_n),\dots,\Com(n,\U_n)$ and once with
  $\Com(n+1,\U_n),\dots,\Com(2n,\U_n)$. For $i\in [T_{\Mest}]$, let $I_{0}^i$ be
  an indicator random variable that takes the value 1 if and only if in the
  $i$-th iteration
  $\Dver(\Com(1,\U_n),\dots,\Com(n,\U_n),\secret_0,\secret_1,\sop)=1$.
  Similarly, denote by $I_{1}^i$ an indicator random variable that takes the
  value 1 if and only if in the $i$-th iteration
  $\Dver(\Com(n+1,\U_n),\dots,\Com(2n,\U_n),\secret_0,\secret_1,\sop)=1$. When
  $\Mest$ finishes, it holds that $q_0 = \sum_{i=1}^T I_0^i$ and $q_1 =
  \sum_{i=1}^T I_1^i$. Furthermore, if $\bias(\secret_0,\secret_1,\sop) \geq
  \eps/3$ we get that $\E[|q_0-q_1|] \geq (\eps/3) \cdot T_{\Mest}$. By
  Chernoff's bound (see \cite[\S A.1]{AlonS92}) we get that
  \begin{align*}
    \Pr[|q_0-q_1| > 3/4 \cdot ((\eps/3) \cdot T_{\Mest})] \ge
    1-\exp\(O(\eps\cdot T_{\Mest})\).
  \end{align*}
  Similarly, if $\bias(\secret_0,\secret_1,\sop) \leq \eps/10$ we get that
  $\E[|q_0-q_1|] \leq (\eps/10) \cdot T_{\Mest}$. By Chernoff's bound we get that
  \begin{align*}
    \Pr[|q_0-q_1| > 2 \cdot ((\eps/10) \cdot T_{\Mest})] \leq
    \exp\(O(\eps\cdot
    T_{\Mest})\).
  \end{align*}

  Recall that $\Mest$ outputs 1 if and only if $|q_0-q_1|> n$. Plugging in
  $T_{\Mest}=4n/\eps$ both parts of the claim follow.
\end{proof}

In \Cref{claim:bias_to_B} we proved that $\Mest$ is a good estimator for the
bias of $\Dver$. That is, we showed that if $\Dver$ is very biased, then $\Mest$
is 1 (with high probability) and vice-versa (\ie that if $\Dver$ is unbiased,
then $\Mest$ is most likely to be 0). Denote by $\BAD$ the event in which
$\Mest(\secret_0,\secret_1,\sop) = 1$ and $\bias(\secret_0,\secret_1,\sop) \leq
\eps/10$. In the next claim we show that the probability that $\BAD$ happens in
any iteration of $D'$ is negligible.

\begin{claim}\label{claim:always_good}
  Denote by $\BAD^i$ the event that $\BAD$ happens in iteration $i\in [T]$.
  \begin{align*}
    \Prp{\forall i:\; \neg\BAD^i} \geq 1-\negl(n).
  \end{align*}
\end{claim}
\begin{proof}
  Since the $T$ iteration are independent and implemented identically it holds
  that
  \begin{align*}
    \Prp{\exists i:\; \BAD^i} & = \sum_{i=1}^T\Prp{\BAD^i}
     = T\cdot \Prp{\BAD}.
  \end{align*}
  Observe that
  \begin{align*}
    \Prp{\BAD} &= \Prp{\Mest(\secret_0,\secret_1,\sop) = 1 \myand
    \bias(\secret_0,\secret_1,\sop) \leq \eps/10} \\
    & \leq \Prp{\Mest(\secret_0,\secret_1,\sop) = 1 \;|\;
    \bias(\secret_0,\secret_1,\sop) \leq \eps/10} \leq \negl(n).
  \end{align*}
  Hence, we get that $\Prp{\exists i:\; \BAD^i} \leq (n/\eps)\cdot
  \negl(n)
  \leq
  \negl(n)$.
\end{proof}

The next claim states that if $\sop$ is such that $M(\sop)=0$, then $B$
outputs $1$ with very high
probability. The idea is to combine
\Cref{claim:mainProofClaim2,claim:mainProofClaim3} that assure that if
$M(\sop)=0$, then $\Dver$ is biased (\ie $\bias$ is large), with
\Cref{claim:bias_to_B} that assures that if the $\bias$ is large, then $\Mest$
almost surely outputs 1.
\begin{claim}\label{lemma:mest_small_diff}
  For $(\secret_0,\secret_1,X)\leftarrow \Samp(1^n)$,
  \begin{align*}
    \Pr\left[\Mest(\secret_0,\secret_1,\sop) = 1 \;|\; M(X)=0 \right] \geq
    1-\negl(n).
  \end{align*}
\end{claim}
\begin{proof}
  Let $(\secret_0,\secret_1,\sop)\leftarrow \Samp(1^n)$. By the definition of
  $\Mest$ it holds that $\Mest(\secret_0,\secret_1,\sop) = 1$ if and only if
  $q_0-q_1 >n$. Thus, it is enough to show that
  \begin{align*}
    \Pr[|q_0-q_1| > n \;|\; M(\sop)=0] \geq 1-\negl(n).
  \end{align*}

  Using \Cref{claim:mainProofClaim2,claim:mainProofClaim3} we get that
  \begin{align*}
    \Pr[\bias(\secret_0,\secret_1,\sop) \geq \eps/2-\negl(n) \;|\; M(\sop)=0]
    \geq 1 -\negl(n).
  \end{align*}
  Plugging this into \Cref{claim:bias_to_B} the claim follows.
\end{proof}

At this point we are finally ready to prove \Cref{lemma:security}.
\begin{proof}[Proof of \Cref{lemma:security}]
  Recall that our goal is to lower bound the following expression:
  \begin{align*}
    |\Pr&[D'(\com_1,\dots,\com_n) = 1 \;|\; Z=A_0] - \Pr[D'(\com_1,\dots,\com_n)
    = 1 \;|\; Z=A_1] |.
  \end{align*}

  Notice that one property of $M$ that follows from the assumption in
  \cref{eq:assumption} is that $\Pr[M(\sop)=0] \geq \eps$ (where the probability
  if over $\Samp$). Combining this fact with the fact that $D'$ makes $T=n/\eps$
  iterations of $\Mest$ and $\Pr\left[\Mest(\secret_0,\secret_1,\sop) = 1 \;|\;
    M(X)=0 \right] \geq 1-\negl(n)$ (by \Cref{lemma:mest_small_diff}), we get
  that $D'$ reaches Step~\ref{item:default} with negligible probability. In
  other words, with probability $1-\negl(n)$ there is an iteration in which
  $\sop$ is chosen such that $M(\sop)=0$ and $\Mest$ outputs 1. For the rest of
  the proof we assume that this is indeed the case (and lose a negligible
  additive term).

  Furthermore, using \Cref{claim:always_good} we may also assume that in every
  iteration $\BAD$ does not happen. That is, in every iteration either $\Mest$
  outputs 0 or $\bias$ is larger than $\eps/10$. Recall that $D'$ ignores all
  the iteration in which $\Mest$ outputs $0$. Moreover, we assumed that there
  is an iteration in which $\Mest$ outputs $1$. In that iteration, it must be
  the case that the $\bias$ is larger than $\eps/10$ which completes the proof.
\end{proof} 
\section{Conclusions and Open Problems}
We have shown a construction of a secret-sharing scheme for any $\mNP$ access
structure. In fact, our construction yields the first candidate computational
secret-sharing scheme for {\em all} monotone functions in $\PP$ (recall that not
every monotone function in $\PP$ can be computed by a polynomial-size monotone
circuit, see \eg Razborov's lower bound for matching \cite{Razborov85}).  Our
construction only requires witness encryption scheme for $\NP$. 
%As we have
%mentioned, a candidate for such a scheme was recently suggested by Garg~\etal
%\cite{GargGSW13}.

We conclude with several open problems:
\begin{itemize}
\item Is there a secret-sharing scheme for $\mNP$  that relies only on standard hardness assumptions, or
  at least falsifiable ones \cite{Naor03}?

\item Is there a way to use secret-sharing for monotone $\PP$ to achieve
  secret-sharing for monotone $\NP$ (in a black-box manner)?

%\item What are the minimal hardness assumption needed to construct a
%  secret-sharing scheme for perfect matching (see Footnote~\ref{fn:matching})?

\item Construct a Rudich secret-sharing scheme for every access
  structure in $\mNP$ that is secure against \emph{adaptive} adversaries (see
  \Cref{sec:adaptive_ss} for a discussion).

  Under a stronger assumption, \ie extractable witness encryption (in which if
  an algorithm is able to decrypt a ciphertext, then it is possible to extract a
  witness), Zvika Brakerski observed that our construction is secure against
  adaptive adversaries as well.

\item Show a completeness theorem (similarly to \Cref{cor:main}) for
  secret-sharing schemes that are also secure against \emph{adaptive}
  adversaries, as defined in \Cref{sec:adaptive_ss}.
\end{itemize}

\section*{Acknowledgements}
  We are grateful to Amit Sahai for suggesting to base our construction on
  witness encryption. We thank Zvika Brakerski for many helpful discussions and
  insightful ideas. The second author thanks Steven Rudich for sharing with him
  his ideas on secret sharing beyond $\PP$. We thank the anonymous referees for
  many helpful remarks.\\\\

\addcontentsline{toc}{section}{References}
\bibliographystyle{alpha}
\bibliography{MainNPSecretSharing}

\newcommand{\etalchar}[1]{$^{#1}$}
\begin{thebibliography}{KMN{\etalchar{+}}14}

\bibitem[AS08]{AlonS92}
Noga Alon and Joel Spencer.
\newblock {\em The Probabilistic Method}.
\newblock John Wiley, third edition, 2008.

\bibitem[Bei11]{Beimel11}
Amos Beimel.
\newblock Secret-sharing schemes: A survey.
\newblock In {\em IWCC}, volume 6639 of {\em Lecture Notes in Computer
  Science}, pages 11--46. Springer, 2011.

\bibitem[BGI{\etalchar{+}}01]{BarakGIRSVY01}
Boaz Barak, Oded Goldreich, Russell Impagliazzo, Steven Rudich, Amit Sahai,
  Salil~P. Vadhan, and Ke~Yang.
\newblock On the (im)possibility of obfuscating programs.
\newblock In {\em CRYPTO}, volume 2139 of {\em Lecture Notes in Computer
  Science}, pages 1--18. Springer, 2001.

\bibitem[BGI{\etalchar{+}}12]{BarakGIRSVY12}
Boaz Barak, Oded Goldreich, Russell Impagliazzo, Steven Rudich, Amit Sahai,
  Salil~P. Vadhan, and Ke~Yang.
\newblock On the (im)possibility of obfuscating programs.
\newblock {\em Journal of the ACM}, 59(2):6, 2012.
\newblock Preliminary version appeared in CRYPTO 2001.

\bibitem[BGK{\etalchar{+}}14]{BarakGKPS14}
Boaz Barak, Sanjam Garg, Yael~Tauman Kalai, Omer Paneth, and Amit Sahai.
\newblock Protecting obfuscation against algebraic attacks.
\newblock In {\em EUROCRYPT}, volume 8441 of {\em Lecture Notes in Computer
  Science}, pages 221--238. Springer, 2014.

\bibitem[BI05]{BeimelI05}
Amos Beimel and Yuval Ishai.
\newblock On the power of nonlinear secrect-sharing.
\newblock {\em SIAM Journal on Discrete Mathematics}, 19(1):258--280, 2005.

\bibitem[BL88]{Leichter88}
Josh~Cohen Benaloh and Jerry Leichter.
\newblock Generalized secret sharing and monotone functions.
\newblock In {\em CRYPTO}, volume 403 of {\em Lecture Notes in Computer
  Science}, pages 27--35. Springer, 1988.

\bibitem[Bla79]{Blakley79}
George~R. Blakley.
\newblock Safeguarding cryptographic keys.
\newblock {\em Proceedings of the AFIPS National Computer Conference},
  22:313--317, 1979.

\bibitem[BR07]{BellareR07}
Mihir Bellare and Phillip Rogaway.
\newblock Robust computational secret sharing and a unified account of
  classical secret-sharing goals.
\newblock In {\em ACM Conference on Computer and Communications Security},
  pages 172--184. ACM, 2007.

\bibitem[BR14a]{BrakerskiR14a}
Zvika Brakerski and Guy~N. Rothblum.
\newblock Black-box obfuscation for d-{CNF}s.
\newblock In {\em ITCS}, pages 235--250. ACM, 2014.

\bibitem[BR14b]{BrakerskiR14}
Zvika Brakerski and Guy~N. Rothblum.
\newblock Virtual black-box obfuscation for all circuits via generic graded
  encoding.
\newblock In {\em TCC}, pages 1--25, 2014.

\bibitem[BZ14]{BonehZ14}
Dan Boneh and Mark Zhandry.
\newblock Multiparty key exchange, efficient traitor tracing, and more from
  indistinguishability obfuscation.
\newblock In {\em CRYPTO (1)}, volume 8616 of {\em Lecture Notes in Computer
  Science}, pages 480--499. Springer, 2014.

\bibitem[DNRS03]{DworkNRS03}
Cynthia Dwork, Moni Naor, Omer Reingold, and Larry~J. Stockmeyer.
\newblock Magic functions.
\newblock {\em Journal of the ACM}, 50(6):852--921, 2003.

\bibitem[GGH{\etalchar{+}}13]{GargGH0SW13}
Sanjam Garg, Craig Gentry, Shai Halevi, Mariana Raykova, Amit Sahai, and Brent
  Waters.
\newblock Candidate indistinguishability obfuscation and functional encryption
  for all circuits.
\newblock In {\em FOCS}, pages 40--49, 2013.

\bibitem[GGSW13]{GargGSW13}
Sanjam Garg, Craig Gentry, Amit Sahai, and Brent Waters.
\newblock Witness encryption and its applications.
\newblock In {\em STOC}, pages 467--476. ACM, 2013.

\bibitem[GLSW14]{GentryLSW14}
Craig Gentry, Allison~B. Lewko, Amit Sahai, and Brent Waters.
\newblock Indistinguishability obfuscation from the multilinear subgroup
  elimination assumption.
\newblock {\em IACR Cryptology ePrint Archive}, 2014:309, 2014.

\bibitem[GLW14]{GentryLW14}
Craig Gentry, Allison~B. Lewko, and Brent Waters.
\newblock Witness encryption from instance independent assumptions.
\newblock In {\em CRYPTO (1)}, volume 8616 of {\em Lecture Notes in Computer
  Science}, pages 426--443. Springer, 2014.

\bibitem[GM84]{GoldwasserM84}
Shafi Goldwasser and Silvio Micali.
\newblock Probabilistic encryption.
\newblock {\em Journal of Computer and System Sciences}, 28(2):270--299, 1984.

\bibitem[GS92]{GrigniS90}
Michelangelo Grigni and Michael Sipser.
\newblock Monotone complexity.
\newblock In {\em Proceedings of {LMS} workshop on Boolean function
  complexity}, volume 169, pages 57--75. Cambridge University Press, 1992.

\bibitem[HILL99]{HastadILL99}
Johan H{\aa}stad, Russell Impagliazzo, Leonid~A. Levin, and Michael Luby.
\newblock A pseudorandom generator from any one-way function.
\newblock {\em SIAM J. Comput.}, 28(4):1364--1396, 1999.

\bibitem[Imp95]{Impagliazzo95}
Russell Impagliazzo.
\newblock A personal view of average-case complexity.
\newblock In {\em Structure in Complexity Theory Conference}, pages 134--147.
  IEEE Computer Society, 1995.

\bibitem[ISN93]{ISN93}
Mitsuru Ito, Akira Saito, and Takao Nishizeki.
\newblock Multiple assignment scheme for sharing secret.
\newblock {\em Journal of Cryptology}, 6(1):15--20, 1993.

\bibitem[KMN{\etalchar{+}}14]{KomargodskiMNPRY14}
Ilan Komargodski, Tal Moran, Moni Naor, Rafael Pass, Alon Rosen, and Eylon
  Yogev.
\newblock One-way functions and (im)perfect obfuscation.
\newblock {\em IACR Cryptology ePrint Archive}, 2014:347, 2014.
\newblock To appear in FOCS 2014.

\bibitem[Kra93]{Krawczyk93}
Hugo Krawczyk.
\newblock Secret sharing made short.
\newblock In {\em CRYPTO}, volume 773 of {\em Lecture Notes in Computer
  Science}, pages 136--146. Springer, 1993.

\bibitem[KW93]{KarchmerW93}
Mauricio Karchmer and Avi Wigderson.
\newblock On span programs.
\newblock In {\em Structure in Complexity Theory Conference}, pages 102--111.
  IEEE Computer Society, 1993.

\bibitem[Nao91]{Naor91}
Moni Naor.
\newblock Bit commitment using pseudorandomness.
\newblock {\em Journal of Cryptology}, 4(2):151--158, 1991.

\bibitem[Nao03]{Naor03}
Moni Naor.
\newblock On cryptographic assumptions and challenges.
\newblock In {\em CRYPTO}, volume 2729 of {\em Lecture Notes in Computer
  Science}, pages 96--109. Springer, 2003.

\bibitem[Nao06]{NaorSlides}
Moni Naor.
\newblock Secret sharing for access structures beyond {P}, 2006.
\newblock Slides:
  \url{http://www.wisdom.weizmann.ac.il/~naor/PAPERS/minicrypt.html}.

\bibitem[PST14]{PassST14}
Rafael Pass, Karn Seth, and Sidharth Telang.
\newblock Indistinguishability obfuscation from semantically-secure multilinear
  encodings.
\newblock In {\em CRYPTO (1)}, volume 8616 of {\em Lecture Notes in Computer
  Science}, pages 500--517. Springer, 2014.

\bibitem[Raz85]{Razborov85}
Alexander~A. Razborov.
\newblock Lower bounds for the monotone complexity of some {B}oolean functions.
\newblock {\em Dokl. Ak. Nauk. SSSR}, 281:798--801, 1985.
\newblock English translation in: \textit{Soviet Math. Dokl.} Vol 31, pp.
  354-357, 1985.

\bibitem[Sha79]{Shamir79}
Adi Shamir.
\newblock How to share a secret.
\newblock {\em Communications of the ACM}, 22(11):612--613, 1979.

\bibitem[SW14]{SahaiW14}
Amit Sahai and Brent Waters.
\newblock How to use indistinguishability obfuscation: deniable encryption, and
  more.
\newblock In {\em STOC}, pages 475--484. ACM, 2014.

\bibitem[VNS{\etalchar{+}}03]{VinodNSRK03}
V.~Vinod, Arvind Narayanan, K.~Srinathan, C.~Pandu Rangan, and Kwangjo Kim.
\newblock On the power of computational secret sharing.
\newblock In {\em INDOCRYPT}, volume 2904 of {\em Lecture Notes in Computer
  Science}, pages 162--176. Springer, 2003.

\end{thebibliography}

\appendix
\section{Proof of Theorem \ref{thm:def_equiv}}\label{sec:def_equiv}

In this section we prove that \Cref{def:rudichSecretSharing} is 
equivalent to \Cref{def:rudichSecretSharingSemanticSecurity}.

\begin{proof}[Proof that \Cref{def:rudichSecretSharingSemanticSecurity} 
implies
  \Cref{def:rudichSecretSharing}] 
  Let $\Sc$ be a Rudich secret-sharing scheme satisfying
  \Cref{def:rudichSecretSharingSemanticSecurity} and assume towards
  contradiction that it does \emph{not} satisfy 
  \Cref{def:rudichSecretSharing}.
  That is, there is a pair of probabilistic polynomial-time algorithms
  $(\Samp,D)$ and a non-negligible $\eps$ such that for
  $(\secret_0,\secret_1,\sop,\sigma) \leftarrow \Samp(1^n)$ it holds 
  that
  \begin{align}\label{eq:indisTOsemsec}
    | \Pr&\left[M(\sop) = 0 \myand D(1^n, \secret_0,\secret_1,
      \shares(\secret_0,\sop), \sigma) = 1 \right] - \\
    & \Pr\left[M(\sop) = 0 \myand D(1^n, \secret_0,\secret_1,
      \shares(\secret_1,\sop), \sigma)= 1 \right] | \geq \eps. \nonumber
   \end{align}
   For a bit $b$ chosen uniformly at random from $\B$, we have that
   \begin{align*}
     \Pr& \left[M(\sop) = 0 \myand D(1^n, \secret_0,\secret_1,
       \shares(\secret_b,\sop), \sigma) = b \right] = \\
     %& \frac12(\Pr\left[M(\sop) = 0 \myand D(1^n, \secret_0,\secret_1,
     %  \shares(\secret_0,\sop), \sigma) = 0 \right] \\
     %& + \Pr\left[M(\sop) = 0 \myand D(1^n, \secret_0,\secret_1,
     %  \shares(\secret_1,\sop), \sigma) = 1 \right]) = \\
     & \frac12(\Pr\left[D(1^n, \secret_0,\secret_1, 
     \shares(\secret_0,\sop),
       \sigma) = 0 \;|\; M(\sop) = 0
     \right]\cdot\Pr[M(\sop) = 0] \\
     & + \Pr\left[M(\sop) = 0 \myand D(1^n, \secret_0,\secret_1,
       \shares(\secret_1,\sop), \sigma) = 1
     \right])= \\
     %& \frac12(\Pr[M(\sop) = 0] - \Pr\left[D(1^n, \secret_0,\secret_1,
     %  \shares(\secret_0,\sop), \sigma) =
     %  1 \;|\; M(\sop) = 0 \right]\cdot\Pr[M(\sop) = 0] \\
     %& + \Pr\left[M(\sop) = 0 \myand D(1^n, \secret_0,\secret_1,
     %  \shares(\secret_1,\sop), \sigma) = 1 \right]) = \\
     & \frac12(\Pr[M(\sop) = 0] - \Pr\left[M(\sop) = 0 \myand D(1^n,
       \secret_0,\secret_1,
       \shares(\secret_0,\sop), \sigma) = 1 \right] \\
     & + \Pr\left[M(\sop) = 0 \myand D(1^n, \secret_0, \secret_1,
       \shares(\secret_1,\sop), \sigma) = 1 \right]).
   \end{align*}
   Plugging in \cref{eq:indisTOsemsec} we get that 
   \begin{align*}
     | \Pr& \left[M(\sop) = 0 \myand D(1^n,  \secret_0,\secret_1, 
       \shares(\secret_b,\sop), \sigma) = b \right] - 1/2\cdot 
       (\Pr[M(\sop)=0]) | \geq
     \eps/2.
   \end{align*}

   Assume that $\Samp$ generates secrets in $[2^t]$ for some $t>0$. Let
   $\mathcal F = \{f_i\colon[2^t]\to\B \;|\; i\in[t] \myand \forall 
   x\in[2^t]:\;
   f_i(x)=\mathsf{bin}(x)_i\}$ be the set of all dictator functions, 
   where
   $\mathsf{bin}(x)$ denotes the binary representation of $x$ of length 
   $t$
   (with leading zeroes if needed).  We define a sampling algorithm 
   $\Samp'$ as
   follows: $\Samp'(1^n)$ first runs $\Samp(1^n)$ and gets two secrets
   $\secret_0,\secret_1$, a subset of parties $\sop$ and auxiliary 
   information
   $\sigma$. Then, $\Samp'$ chooses a bit $b\in\B$ uniformly at random 
   and
   outputs $(\secret_b,\sop,\sigma')$, where $\sigma' = \langle
   \secret_0,\secret_1,\sigma \rangle$. The algorithm $D'$ emulates the
   execution of $D$ with inputs $\secret_0,\secret_1$, 
   $\shares(S_b,\sop)$ and
   $\sigma'$. Note that $D'$ does not know the bit $b$. Denote by 
   $\mathcal
   F'\subseteq \mathcal F$ the set of function $f\in\mathcal F$ for 
   which
   $f(S_0) \neq f(S_1)$. Observe that with probability strictly larger 
   than 0
   over a random choice of $f$ from $\mathcal F$ it holds that $f\in 
   \mathcal
   F'$ (\ie $\mathcal F'$ is not empty). Then, over the randomness of 
   $\Samp'$
   we have that for any $f\in\mathcal F'$
   \begin{align}\label{eq:equivalence1}
     | \Pr&\left[M(\sop) = 0 \myand D'(1^n, 
     \shares(\secret_b,\sop),\sigma') =
       f(\secret_b) \right] - 1/2\cdot\Pr[M(\sop)=0]| \geq \eps/2.
   \end{align}
   On the other hand, since $\sop$ does not have any information about
   $\secret_0,\secret_1$ and $b$ is chosen uniformly at random from 
   $\B$, for
   any algorithm $D''$ and every $f\in \mathcal F'$ it holds that
   \begin{align*}
     \Pr\left[D''(1^n,\sop,\sigma') = f(\secret_b)\right] = 1/2.
   \end{align*}
   Thus,
   \begin{align}\label{eq:equivalence2}
     \Pr\left[M(\sop)=0 \myand D''(1^n,\sop,\sigma') = 
     f(\secret_b)\right] =
     1/2\cdot \Pr[M(\sop)=0].
   \end{align}
   Combining \cref{eq:equivalence1,eq:equivalence2} we get that for any
   $f\in\mathcal F'$:
   \begin{align*}
     | \Pr&\left[M(\sop) = 0 \myand D'(1^n, 
     \shares(\secret_b,\sop),\sigma')
       = f(\secret_b) \right]  - \\
     & \Pr\left[M(\sop) = 0 \myand D''(1^n,\sop,\sigma') = f(\secret_b) 
     \right]
     | \geq \eps/2
   \end{align*}
   which contradicts the unlearnability requirement of
   \Cref{def:rudichSecretSharingSemanticSecurity}.
 \end{proof}

\begin{proof}[Proof that \Cref{def:rudichSecretSharing} implies
  \Cref{def:rudichSecretSharingSemanticSecurity}]
  
  Let $\Sc$ be a Rudich secret-sharing scheme satisfying
  \Cref{def:rudichSecretSharing}. Fix a pair of algorithms $(\Samp,D)$ 
  and a
  function $f$ as in \Cref{def:rudichSecretSharingSemanticSecurity}. We 
  define 
  a
  simulator $D'$ as follows:
  \begin{align*}
    D'(1^n,\sop,\sigma)=D(1^n,\shares(0,\sop),\sigma).
  \end{align*}
  We prove that this simulator satisfies the \emph{unlearnability of 
  the 
  secret}
  requirement in \Cref{def:rudichSecretSharingSemanticSecurity}. 
  Namely, we 
  show
  that
  \begin{align*}
    | \Pr&[M(\sop)=0 \myand 
    D(1^n,\shares(\secret,\sop),\sigma)=f(\secret)] -
    \\
    &\Pr[M(\sop)=0 \myand D'(1^n,\sop,\sigma)=f(\secret)] | \leq 
    \negl(n).
  \end{align*}

  Towards this end, assume towards contradiction that there exists a
  non-negligible $\eps=\eps(n)$ such that
  \begin{align*}
    | \Pr&[M(\sop)=0 \myand 
    D(1^n,\shares(\secret,\sop),\sigma)=f(\secret)] -
    \\ 
    &\Pr[M(\sop)=0 \myand D'(1^n,\sop,\sigma)=f(\secret)] | \geq  \eps.
  \end{align*}
  Plugging in the definition of $D'$ we have that
  \begin{align*}
    | \Pr&[M(\sop)=0 \myand 
    D(1^n,\shares(\secret,\sop),\sigma)=f(\secret)] -
    \\
    & \Pr[M(\sop)=0 \myand D(1^n,\shares(0,\sop),\sigma)=f(\secret)] | 
    \geq
    \eps.
  \end{align*}
  Next, we define a pair of algorithms $(\Samp'',D'')$ that are good
  distinguishers between two secrets which, in turn, contradicts the
  \emph{indistinguishability of the secret} requirement from
  \Cref{def:rudichSecretSharing} that $\Sc$ satisfies. The sampling 
  algorithm
  $\Samp''$ simply runs $\Samp$ to get $(\secret,\sop,\sigma)$ and 
  output
  $(0,\secret,\sop,\sigma)$. The distinguisher $D''$ is defined as 
  follows: For
  every $b\in\B:\ 
  D''(1^n,\secret_0,\secret_1,\shares(\secret_b,\sop),\sigma)=1$
  if and only if $D(1^n,\shares(\secret_b,\sop),\sigma)=f(\secret_1)$. 
  Using
  this $D''$ we get that
  \begin{align*}
    | \Pr&[M(\sop)=0 \myand
    D''(1^n,\secret_0,\secret_1,\shares(\secret_1,\sop),\sigma)=1] - \\
    &\Pr[M(\sop)=0 \myand
    D''(1^n,\secret_0,\secret_1,\shares(\secret_0,\sop),\sigma)=1] 
    |\geq \eps,
  \end{align*}
  which contradicts the indistinguishability assumption.
\end{proof}
\section{Proof of Lemma \ref{lemma:list_of_hardness}}\label{sec:prf_lemma_lest}
In this section we prove the following lemma.
\begin{replemma}{lemma:list_of_hardness}[Restated]
\ManytoOne
\end{replemma}

\begin{proof}
  Assume that there exists a polynomial-time algorithm $D$ and some 
  $\eps =
  \eps(n)$ such that
  \begin{align}\label{eq:list_ind}
    | \Pr& [D(\Com(1,\U_{n}),\dots,\Com(n,\U_{n}))=1] - \\
    & \Pr[D(\Com(n+1,\U_{n}),\dots, \Com(2n,\U_{n}))=1] | \ge
    \eps. \nonumber
  \end{align}
  For $\sigma\in[2n]$ let $\com_\sigma$ be a random variable sampled 
  according
  to the distribution $\Com(\sigma,\U_{n})$. With this notation,
  \cref{eq:list_ind} can be rewritten as
  \begin{align}\label{eq:list_ind3}
    \abs{ \Pr[D(\com_{1},\dots,\com_{n})=1] -
      \Pr_{}[D(\com_{n+1},\dots,\com_{2n})=1] } \geq \eps.
  \end{align}

  For $1\leq i \leq n-1$ let $\mathcal C^{(i)}$ be the distribution 
  induced by
  the sequence $\com_1,\dots,\com_{n-i},\com_{2n-i+1},\dots,\com_{2n}$.
  Moreover, let $\mathcal C^{(0)}$ be the distribution 
  $\com_1,\dots,\com_n$ 
  and
  let $\mathcal C^{(n)}$ be the distribution 
  $\com_{n+1},\dots,\com_{2n}$. 
  Using
  this notation, \cref{eq:list_ind3} can be rewritten as
  \begin{align*}
    \abs{ \Pr[D(\mathcal{C}^{(0)})=1] - \Pr_{}[D(\mathcal{C}^{(k)})=1] 
    } \geq
    \eps.
  \end{align*}
  By a hybrid argument, there exists an index $i\in [n]$ for which
  \begin{align*}
    \abs{ \Pr[D(\mathcal{C}^{(i-1)})=1] - 
    \Pr_{}[D(\mathcal{C}^{(i)})=1] } \geq
    \eps/n.
  \end{align*}
  Expanding the definition of $\mathcal{C}^{(i)}$, 
  \begin{align*}
    | \Pr&[D\(\com_1,\dots,\com_{n-i},\com_{n-i+1},\com_{2n-i+2},\dots ,
    \com_{2n}\) =1
    ] -   \\
    & \Pr[D(\com_1,\dots,\com_{n-i},\com_{2n-i+1},\com_{2n-i+2},\dots,
    \com_{2n})= 1] | \geq \eps/n.
  \end{align*}
  
  At this point, it follows that there exists $D'$ that distinguishes 
  between
  $\com_{n-i+1}$ and $\com_{2n-i+1}$. Namely, for $x=n-i+1$ and 
  $y=2n-i+1$, it
  holds that
  \begin{align*}
    \abs{ \Pr[D'(\Com(x,\U_{n}))=1] - \Pr[D'(\Com(y,\U_{n}))=1] } \ge 
    \eps/n,
  \end{align*}  
  as required.
\end{proof}
\newpage
\section{Secret-Sharing Zoo}
A summary of the known relations between secret-sharing and other objects.

\newcommand{\sign}{{$^\dagger$}}
\newcommand{\signa}{{$^\ddagger$}}

\begin{figure}[h!]\label{fig:zoo}
\centering
\tikzset{
	%Define standard arrow tip
	>=stealth',
	%Define style for boxes
	punkt/.style={
		rectangle,
		rounded corners,
		draw=black, very thick,
		text width=6.5em,
		minimum height=2em,
		text centered},
	% Define arrow style
	pil/.style={
		->,
		thick,
		shorten <=2pt,
		shorten >=2pt,}
}

	\begin{tikzpicture}[node distance=1cm, auto,]
	%nodes
	\node[] (dummy) {};
	
	\node[punkt,right=6cm of dummy,ultra thick,draw=red] (ssnp)
        {Secret-sharing for $\NP$};
	
	\node[punkt, below=1.5cm of ssnp](ssp) {Secret-sharing for $\PP$}
  	 edge[pil,<-, bend right=25] node[auto] {} (ssnp);
	
         \node[punkt, below=6cm of dummy](ssmp) {Secret-sharing for monotone
           circuit in $\PP$} edge[pil,<-, bend right=25] node[auto] {} (ssp);

         \node[punkt,above=3cm of dummy] (io) {iO for $\PP$};
	
	\node[punkt,left=2cm of dummy] (witness) {Witness encryption}
  	 edge[pil,<-, bend left=35] node[auto] {\cite{GargGH0SW13}} (io)
   	 edge[pil,<-, bend right=15] node[below] {\cite{GargGSW13}} (ssnp)
   	 edge[pil,->, ultra thick, bend left=15,draw=red] node[above] {This work           \sign} (ssnp); 
	
	\node[punkt,below=1.5cm of witness] (oblivious) {Oblivious transfer}
  	 edge[pil,<-, bend right=15] node[below,rotate=10] {[Rudich90] \sign} (ssnp)
   	 edge[pil,<-, bend left=25] node[auto] {\cite{SahaiW14} \sign} (witness);
	
	\node[punkt,below=3.2cm of dummy](owf) {One-way functions}
  	 edge[pil,->, bend right=25] node[auto] {[Yao89]} (ssmp)
   	 edge[pil,<-, bend left=20] node[right] {\cite{KomargodskiMNPRY14}
           \signa} (io); 
  
	\end{tikzpicture}
	\vspace{1em}
        \caption{Secret-sharing Zoo. A {\sign} mark on a line denotes the fact
          that the reduction between the primitives relies also on the existence
          of one-way fucntions. iO stands for indistinguishability
          obfusaction. A {\signa} mark on the line from iO for $\PP$ to one-way
          functions means that the reduction assumes a worst-case complexity
          assumption, namely, that $\NP \not\subseteq \ioBPP$ (see
          \cite{KomargodskiMNPRY14} for more information). [Yao89] and
          [Rudich90] are unpublished.} \label{fig:M1}

\end{figure}
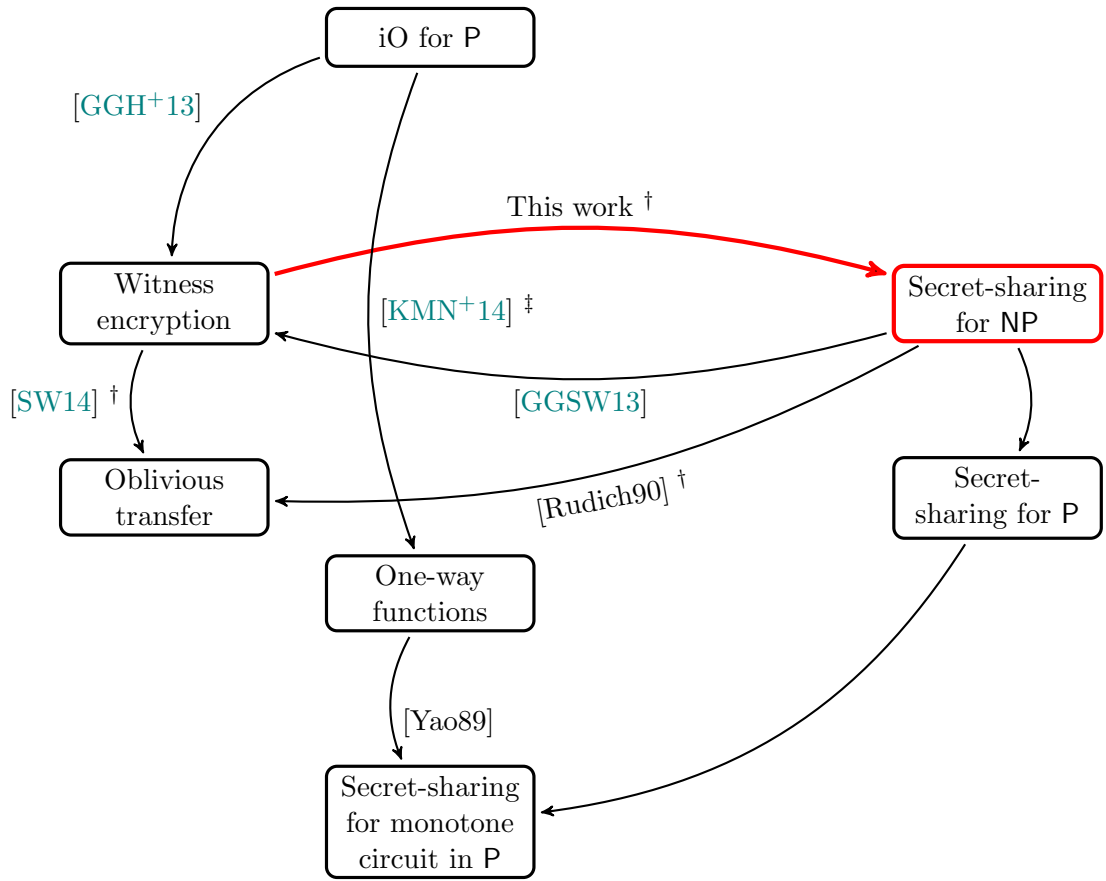

%%% Local Variables:
%%% mode: latex
%%% TeX-master: t
%%% End:

%\input{SSFromReduction}

\end{document}